\documentclass[epjST]{svjour}
\usepackage{amsfonts, amsmath, amssymb, xfrac}
\usepackage{enumitem}
\setlist{noitemsep,topsep=0pt,parsep=0pt,partopsep=0pt}
\usepackage{array}
\usepackage[LCY]{fontenc}
\fontencoding{LCY}\selectfont
\usepackage[cp1251]{inputenc}
\usepackage[english]{babel}
\usepackage[margin=10pt,parskip=5pt,indention=10pt, lofdepth,lotdepth]{subfig}
\usepackage{listings, xcolor}
\usepackage{bm}
\usepackage{cite}
\usepackage[suffix=, outdir=./]{epstopdf}

\usepackage{graphicx}
\usepackage{caption}

\newcommand{\mi}{\mathrm{i}}

\newcommand{\beginMatlab}[1]
{
	\definecolor{dkgreen}{rgb}{0,0.6,0}
	\definecolor{gray}{rgb}{0.5,0.5,0.5}
	\lstloadlanguages{Matlab}%
	\lstset{
	language=Matlab,
	keywordstyle=[1]\color{blue}\bfseries, 
	keywordstyle=[2]\color{dkgreen},   
	keywordstyle=[3]\color{black},      
	escapebegin=\color{dkgreen},
	keywords={abs,collect,simple,exp,log,eigen,break,case,catch,continue,else,elseif,end,for,
	function,global,if,otherwise,persistent,return,switch,try,while,diff,simplify,
	subs,int2str,syms,clear,all, factor,hold, on, off, plot,solve, title,xlabel, ylabel,
	grid,eval},
	caption = {#1},
	basicstyle=\ttfamily,
	morekeywords={},
	commentstyle=\color{dkgreen},
	stringstyle=\color{gray},
	numbers=left, 
	firstnumber=1, 
	numberstyle=\tiny\color{blue}, 
	numbersep=10pt,
	backgroundcolor=\color{white},
	frame=shadowbox,
	tabsize=4,
	showspaces=false,
	showstringspaces=false,
	numbers=left,
	mathescape
 	}
}

\DeclareMathOperator{\LE}{LE}
\DeclareMathOperator{\LCE}{LCE}

\begin{document}
\title{
Homoclinic orbits, and self-excited and hidden attractors in a Lorenz-like system describing convective fluid motion
}
\subtitle{Homoclinic orbits, and self-excited and hidden attractors}
\author{
G. A. Leonov\inst{1} \and
N. V. Kuznetsov\inst{\,1}\inst{2}\fnmsep\thanks{
Corresponding author: nkuznetsov239@gmail.com (N.V. Kuznetsov)
}
\and
T. N. Mokaev\inst{1}\inst{2}
}
\institute{
Faculty of Mathematics and Mechanics,
St. Petersburg State University, St. Petersburg, Russia \and
Department of Mathematical Information Technology,
University of Jyv\"{a}skyl\"{a}, Jyv\"{a}skyl\"{a}, Finland
}

\abstract{
In this paper, we discuss self-excited and hidden attractors for systems of differential equations.
We considered the example of a Lorenz-like system
derived from the well-known Glukhovsky--Dolghansky and Rabinovich systems,
to demonstrate the analysis of self-excited and hidden attractors
and their characteristics.
We applied the \emph{fishing principle} to demonstrate the existence of a homoclinic orbit,
proved the dissipativity and completeness of the system, and found absorbing and positively invariant sets.
We have shown that this system has a self-excited attractor and a {\it hidden} attractor for certain parameters.
The upper estimates of the Lyapunov dimension of self-excited and hidden attractors were obtained analytically.
}
\maketitle

\section{Introduction: self-excited and hidden attractors}
\label{intro}
 When the theories of dynamical systems and oscillations were first developed
 (see, e.g., the fundamental works of Poincare and Lyapunov), researchers mainly
 focused on analyzing equilibria stability and the birth of periodic oscillations.
 The structures of many applied systems (see, e.g., the Rayleigh \cite{Rayleigh-1877},
 Duffing \cite{Duffing-1918}, van der Pol \cite{VanDerPol-1926}, Tricomi \cite{Tricomi-1933},
 and Beluosov-Zhabotinsky \cite{Belousov-1959} systems) are such
 that it is almost obvious that periodic oscillations exist, because the oscillations
 are excited by an unstable equilibrium.
 This meant that scientists of that time could compute such oscillations
 (called self-excited oscillations) by constructing a solution using initial data from a
 small neighborhood of the equilibrium, observing how it is attracted, and visualizing
 the oscillation (\emph{standard computational procedure}).
 In this procedure, computational methods and the engineering notion of a
 \emph{transient process} were combined to study oscillations.

 At the end of the 19th century Poincare considered Newtonian dynamics of the three body problem,
 and revealed the possibility of more complicated behaviors of orbits
 ``{\it so tangled that I cannot even begin to draw them}''.
 He arrived at the conclusion that ``{\it it may happen that small differences in the initial positions
  may lead to enormous differences in the final phenomena}''.
 Further analyses and visualizations of such complicated ``chaotic'' systems became possible
 in the middle of the 20th century after the appearance of powerful computational tools.

 An oscillation can generally be easily numerically localized if the initial data from its open
 neighborhood in the phase space (with the exception of a minor set of points) lead to
 a long-term behavior that approaches the oscillation.
 From a computational perspective, such an oscillation (or set of oscillations)
 is called an attractor, and its attracting set is called the basin of attraction
 (i.e., a set of initial data for which the trajectories tend to the attractor).

 The first well-known example of a visualization of chaotic behavior
 in a dynamical system is from the work of Lorenz \cite{Lorenz-1963}.
 It corresponds to the excitation of chaotic oscillations
 from unstable equilibria, and could have been found using the standard computational procedure
 (see Fig.~\ref{fig:lorenz:attr:se}).
 Later, various self-excited chaotic attractors
 were discovered in many continuous and discrete systems
 (see, e.g., \cite{Rossler-1976,Henon-1976,ChuaKM-1986,Sprott-1994,Celikovsky-1994,ChenU-1999,LuChen-2002}).

\begin{figure}[!ht]
 \centering
 \subfloat[
 {\scriptsize Initial data near the equilibrium $S_0$}
 ] {
 \label{fig:lorenz:attr:se0}
 \includegraphics[width=0.3\textwidth]{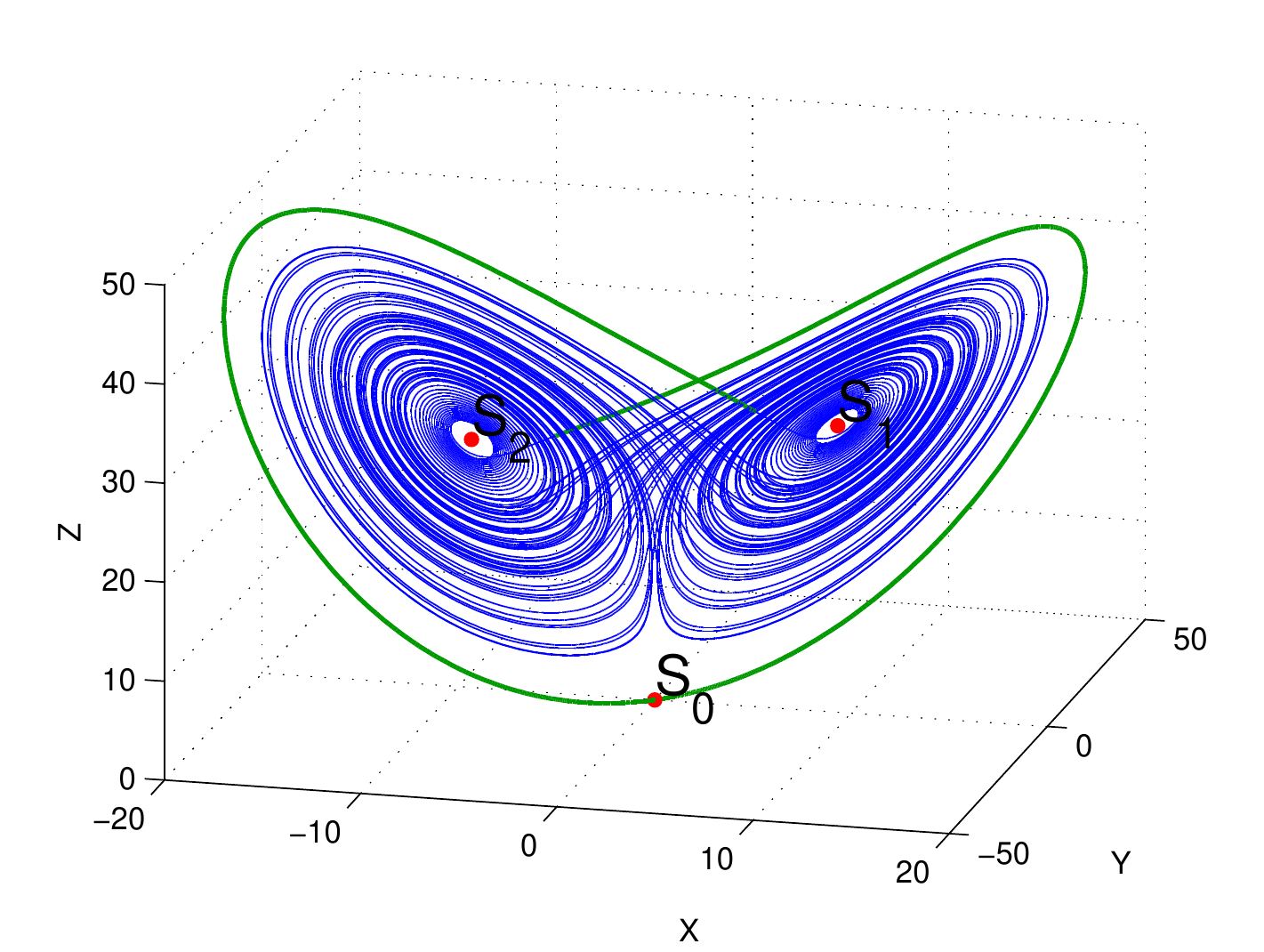}
 }~
 \subfloat[
 {\scriptsize Initial data near the equilibrium $S_1$}
 ] {
 \label{fig:lorenz:attr:se1}
 \includegraphics[width=0.3\textwidth]{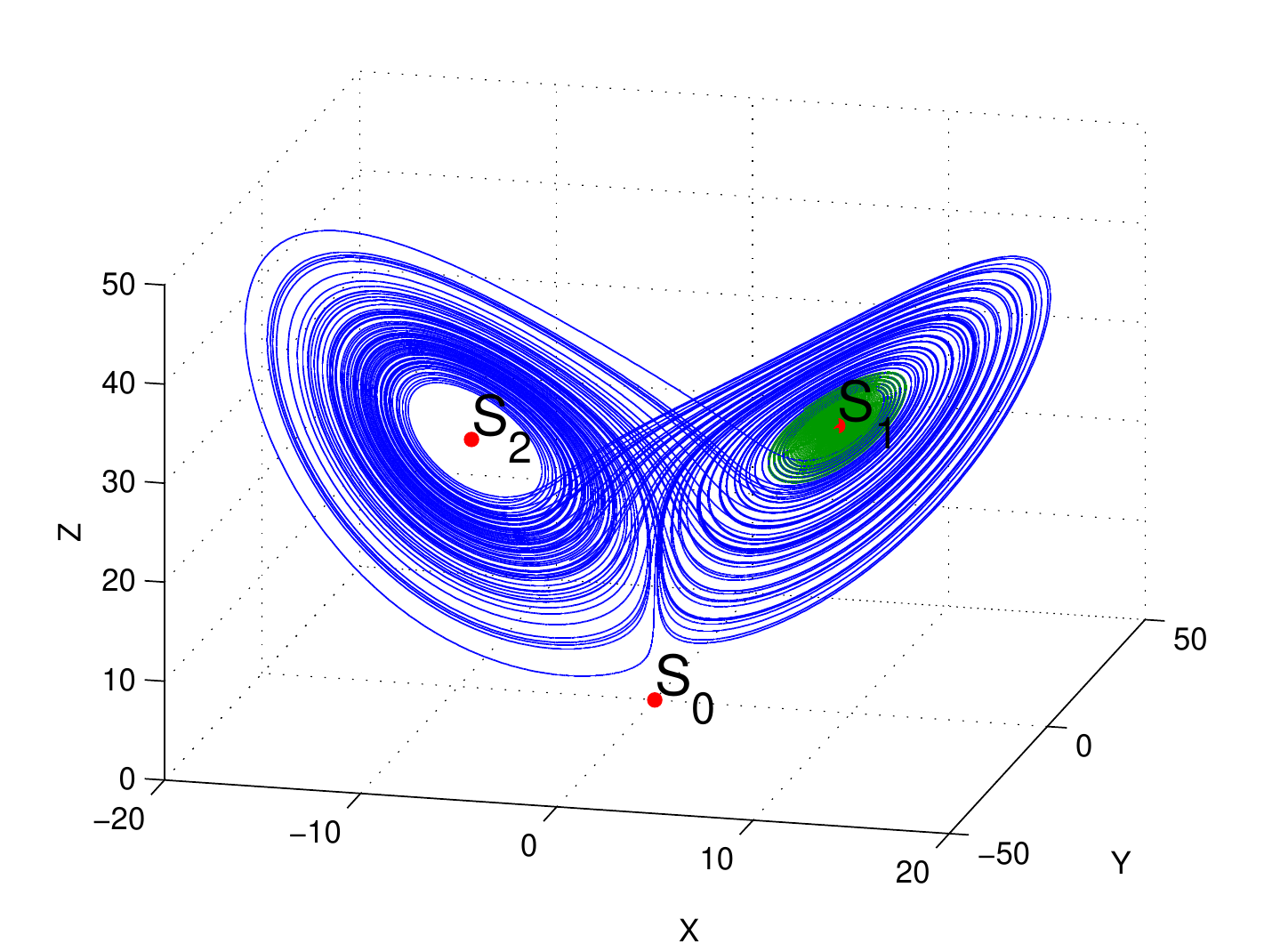}
 }~
 \subfloat[
 {\scriptsize Initial data near the equilibrium $S_2$}
 ] {
 \label{fig:lorenz:attr:se2}
 \includegraphics[width=0.3\textwidth]{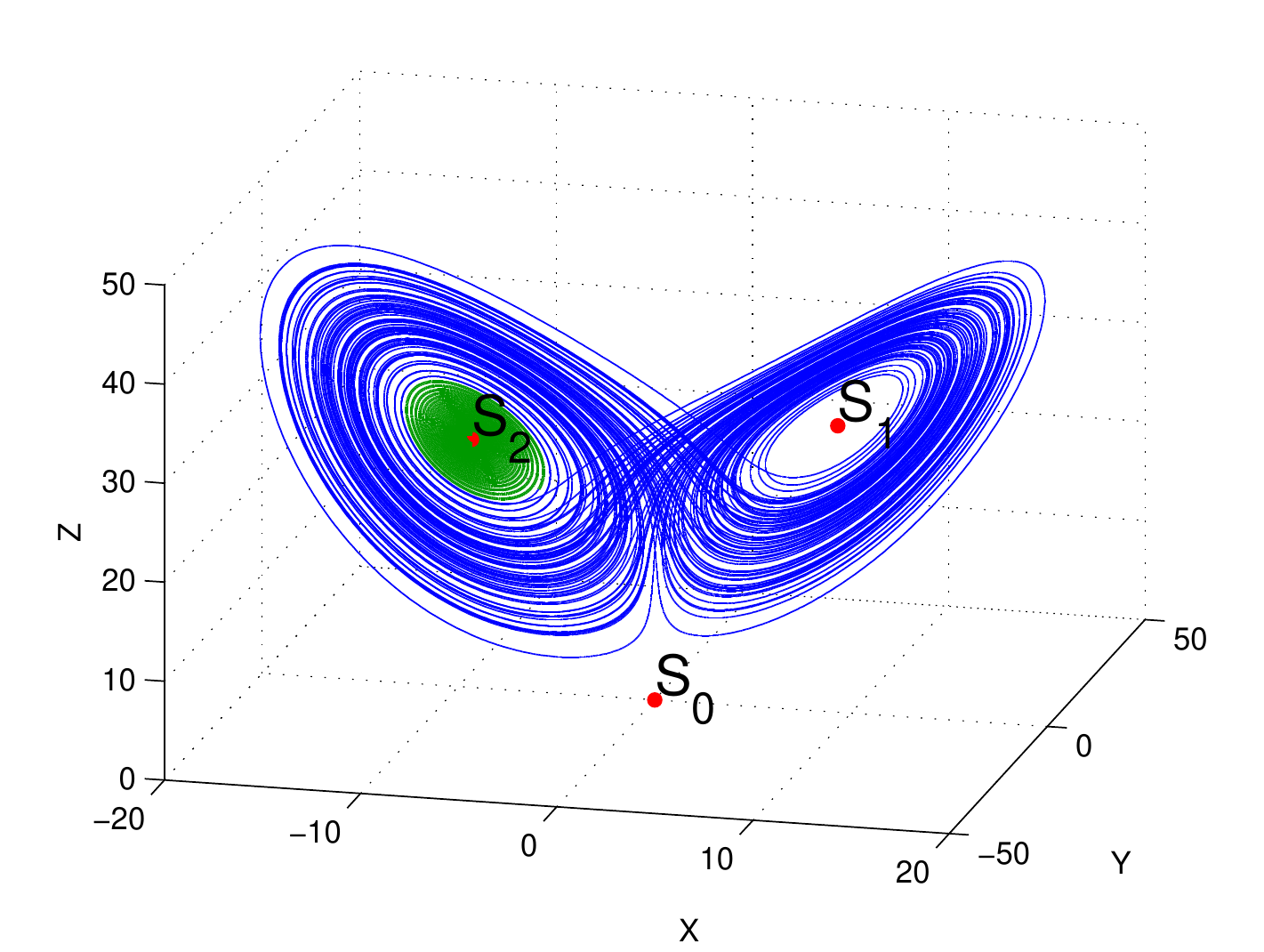}
 }
 \caption{
 Numerical visualization of
 the classical, self-excited, chaotic attractor in the Lorenz system
 $
 \dot x= 10(y-x),
 \dot y= 28 x-y-xz,
 \dot z=-8/3 z+xy
 $
 by the trajectories that start in small neighborhoods of the unstable equilibria $S_{0,1,2}$.
 Here the separation of the trajectory into transition process (green)
 and approximation of attractor (blue) is rough.
 }
 \label{fig:lorenz:attr:se}
\end{figure}

The study of an autonomous (unperturbed) system typically begins with an analysis of the equilibria,
which are easily found numerically or analytically.
Therefore, from a computational perspective, it is natural to suggest the following classification
of attractors \cite{KuznetsovLV-2010-IFAC,LeonovKV-2011-PLA,LeonovKV-2012-PhysD,LeonovK-2013-IJBC},
which is based on the simplicity of finding their basins of attraction in the phase space:

\begin{definition}
\cite{KuznetsovLV-2010-IFAC,LeonovKV-2011-PLA,LeonovKV-2012-PhysD,LeonovK-2013-IJBC}
 An attractor is called a \emph{self-excited attractor}
 if its basin of attraction
 intersects with any open neighborhood of a stationary state (an equilibrium),
 otherwise it is called a \emph{hidden attractor}.
\end{definition}


The basin of attraction for a hidden attractor is not connected with any equilibrium.
For example, hidden attractors are attractors in systems
with no equilibria or with only one stable equilibrium
(a special case of the multistability: coexistence of attractors in multistable systems).
Note that multistability can be inconvenient in various practical applications
(see, for example, discussions on problems related
to the synchronization of coupled multistable systems
in \cite{Kapitaniak-1992,Kapitaniak-1996,PisarchikF-2014,KuznetsovL-2014-IFACWC}).
Coexisting self-excited attractors can be found using
the standard computational procedure\footnote{
We have not discussed such possible computational difficulties as caused by Wada and riddled basins},
whereas there is no standard way of predicting the existence
or coexistence of hidden attractors in a system.

Hidden attractors arise in connection with various
fundamental problems and applied models.
The problem of analyzing hidden periodic oscillations first
arose in the second part of Hilbert's 16th problem (1900),
which considered the number and mutual disposition of limit cycles
in two-dimensional polynomial systems \cite{Hilbert-1901}.
The first nontrivial results were obtained by Bautin (see, e.g., \cite{Bautin-1952}),
which were devoted to the theoretical construction
of three nested limit cycles around one equilibrium in quadratic systems.
Bautin's method can only be used to construct nested, small-amplitude limit cycles,
which can hardly be visualized.
However, recently an analytical approach has been developed, which can be used to
effectively visualize nested, normal amplitude limit cycles in quadratic systems
\cite{LeonovK-RCD-2010,LeonovK-2013-IJBC,KuznetsovKL-2013-DEDS}.

\begin{figure}[!ht]
 \centering
 \subfloat[
 {\scriptsize }
 ] {
 \label{fig:1-lim-cycle}
 \includegraphics[width=0.5\textwidth]{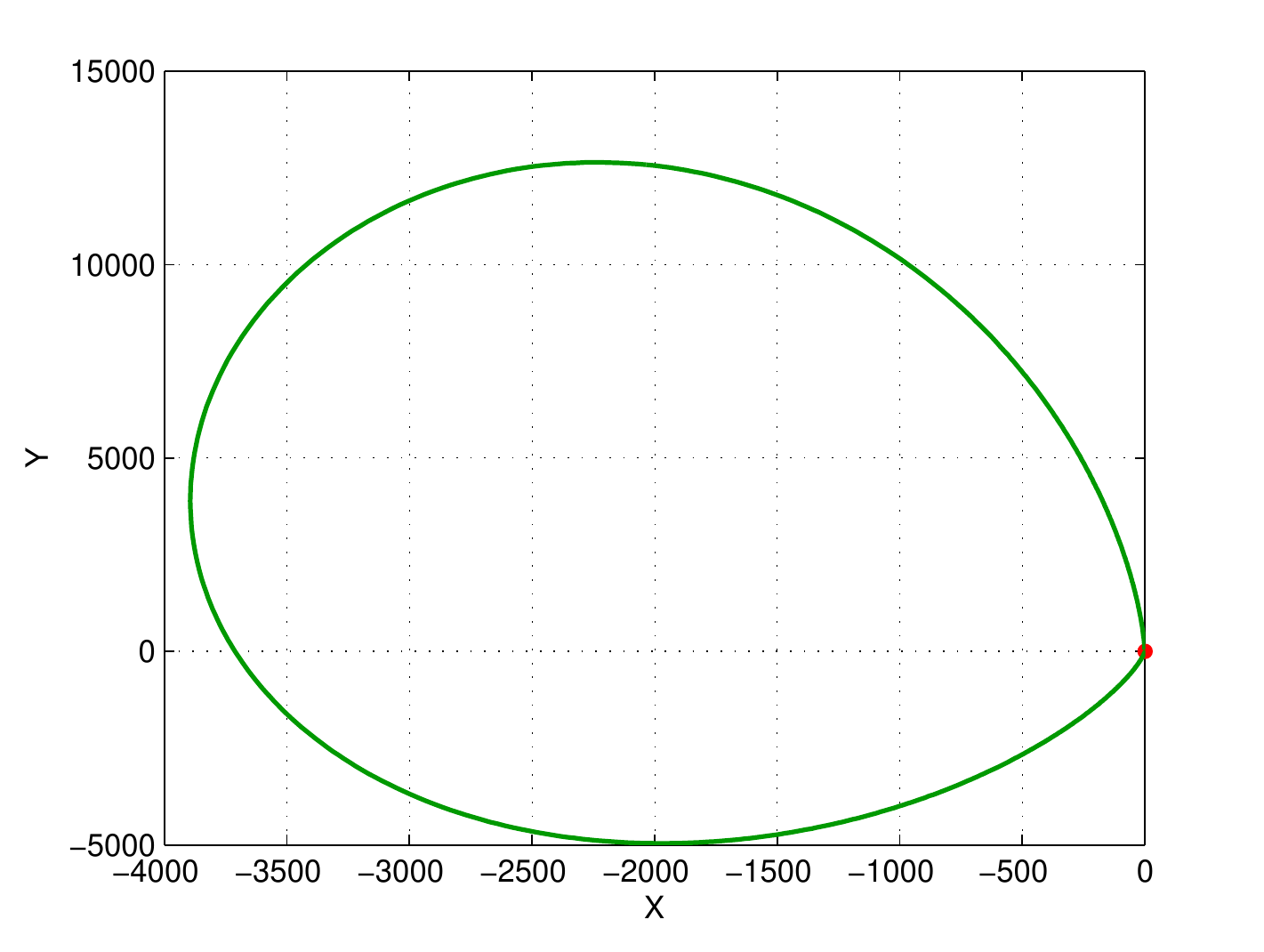}
 }~
 \subfloat[
 {\scriptsize }
 ] {
 \label{fig:3-lim-cycle}
 \includegraphics[width=0.5\textwidth]{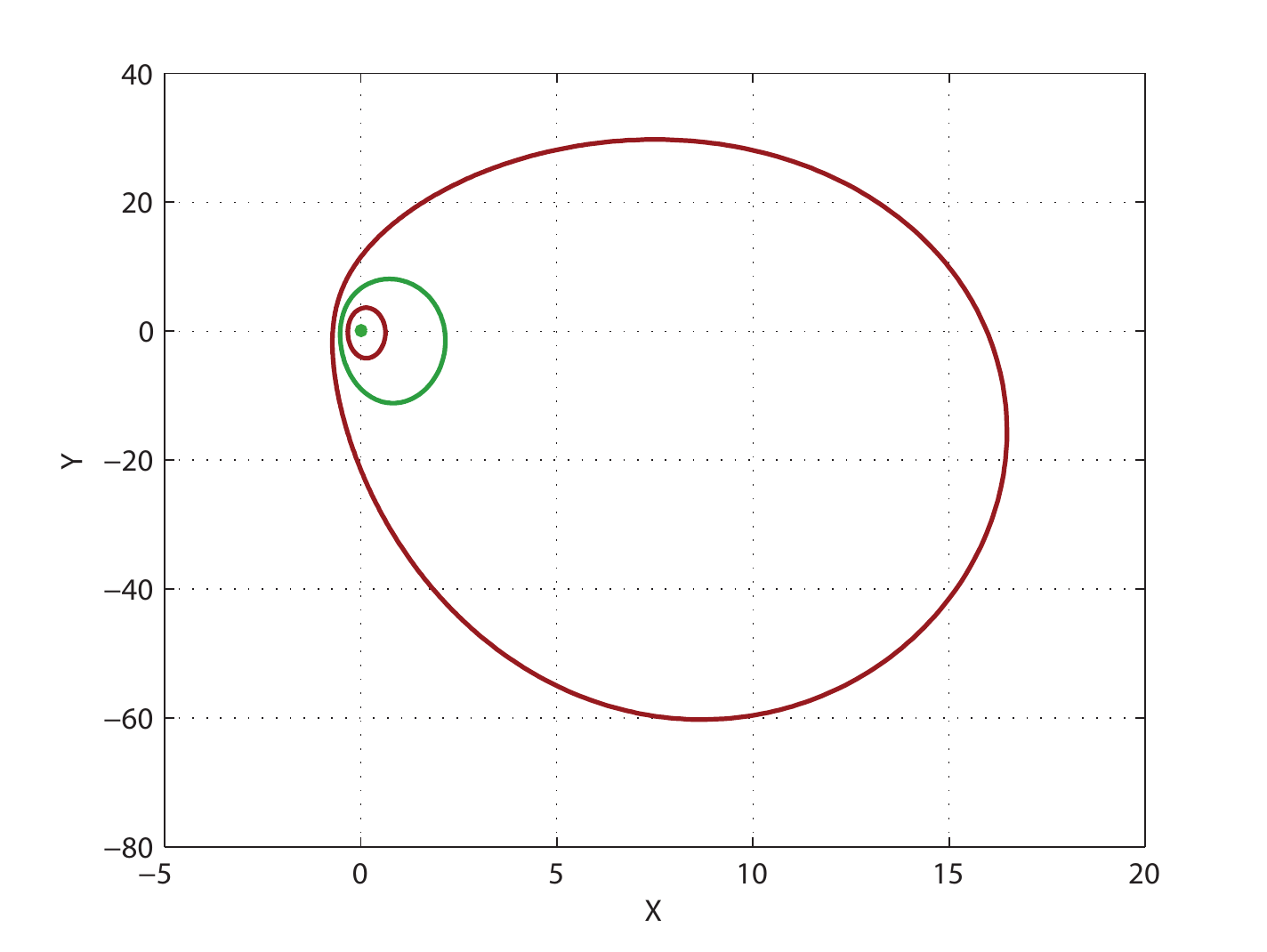}
 }
 \caption{
 Visualization of four limit cycles (green represents stable and red represents unstable)
 in a two-dimensional polynomial quadratic system
 $
 \dot{x} = -(a_1 x^2 + b_1 x y + c_1 y^2 + \alpha_1 x + \beta_1 y),
 \dot{y} = -(a_2 x^2 + b_2 x y + c_2 y^2 + \alpha_2 x + \beta_2 y),
 $
 for the coefficients $a_1 = b_1 = \beta_1 = -1$, $c_1 = \alpha_1 = 0$,
 $b_2 = -2.2$, and $c_2 = -0.7, a_2 = 10$, $\alpha_2 = 72.7778$,
 and $\beta_2 = -0.0015$.
 Localization of three nested limit cycles around the stable zero point (green dot) and
 one limit cycle to the left of the straight line $x=-1$.
 }
 \label{fig:4-lim-cycles}
\end{figure}

Later, in the 1950s-1960s, studies of the well-known
Markus-Yamabe's \cite{MarkusY-1960},
Aizerman's \cite{Aizerman-1949}, and Kalman's \cite{Kalman-1957}
conjectures on absolute stability
led to the discovery of the possible coexistence of
a hidden periodic oscillation and a unique stable stationary point
in automatic control systems
(see \cite{Pliss-1958,Fitts-1966,Barabanov-1988,BernatL-1996,BraginVKL-2011,
LeonovK-2011-DAN,LeonovK-2011-IFAC,KuznetsovLS-2011-IFAC};
the corresponding discrete examples were considered in \cite{Alli-Oke-2012-cu}).

The Rabinovich system \cite{Rabinovich-1978} and the Glukhovsky-Dolghansky system \cite{GlukhovskyD-1980}
are among the first known chaotic systems that have hidden chaotic attractors  \cite{LeonovKM-2015-CNSNS,KuznetsovLM-2015}.
The first one describes the interaction of plasma waves and was considered in 1978 by
Rabinovich \cite{Rabinovich-1978,PikovskiRT-1978}
Another is a model of convective fluid motion and was considered in 1980
by Glukhovsky and Dolghansky \cite{GlukhovskyD-1980}
(which we consider in the remainder of this paper).

Hidden oscillations appear naturally in systems without equilibria, describing
various mechanical and electromechanical models with rotation, and electrical circuits with cylindrical phase space.
One of the first examples is from a 1902 paper \cite{Sommerfeld-1902}
in which Zommerfield analyzed the vibrations caused by a motor driving an unbalanced weight
and discovered the so-called Zommerfield effect (see, e.g., \cite{BlekhmanIF-2007,Eckert-2013}).
Another well-known chaotic system without equilibria
is the Nos\`{e}--Hoover oscillator \cite{Nose-1984,Hoover-1985}
(see also the corresponding Sprott system, which was
discovered independently \cite{Sprott-1994,SprottHH-2014}).
In 2001, a hidden chaotic attractor was reported in a power system with no equilibria \cite{Venkatasubramanian-2001}
(and references within).


After the idea of a ``hidden attractor'' was introduced and
the first hidden Chua attractor was discovered
\cite{LeonovK-2009-PhysCon,KuznetsovLV-2010-IFAC,KuznetsovVLS-2011,LeonovKV-2011-PLA,KuznetsovKLV-2011-ICINCO,LeonovKV-2012-PhysD,KuznetsovKLV-2013},
hidden attractors have received much attention.
Results on the study of hidden attractors were presented
in a number of invited survey and plenary lectures
at various international conferences\footnote{
X Int. Workshop on
Stability and Oscillations of Nonlinear Control Systems (Russia, 2008),
Physics and Control \cite{LeonovK-2009-PhysCon} (Italy, 2009),
3rd International Conference on Dynamics, Vibration and Control
(Hangzhou, China, 2010),
IFAC 18th World Congress \cite{LeonovK-2011-IFAC} (Italy, 2011),
IEEE 5th Int. Workshop on Chaos-Fractals Theories and Applications \cite{LeonovK-2012-IEEE}
(Dalian, China, 2012),
International Conference on Dynamical Systems and Applications (Ukraine, 2012),
Nostradamus (Czech Republic, 2013) \cite{LeonovK-2013-AISC} and others.}.
In 2012,
an invited comprehensive survey on hidden attractors was prepared
for the International Journal of Bifurcation and Chaos \cite{LeonovK-2013-IJBC}.

Many researchers are currently
studying hidden attractors.
Hidden periodic oscillations and hidden chaotic attractors
have been studied in models such as phase-locked loops
\cite{KuznetsovLYY-2014-IFAC,KuznetsovKLNYY-2015-ISCAS},
Costas loops \cite{BestKKLYY-2015-ACC}, drilling systems
\cite{KiselevaKLN-2012-IEEE,LeonovKKSZ-2014},
DC-DC converters \cite{ZhusubaliyevM-2015-HA},
aircraft control systems \cite{AndrievskyKLP-2013-IFAC},
launcher stabilization systems \cite{AndrievskyKLS-2013-IFAC},
plasma waves interaction \cite{KuznetsovLM-2015},
convective fluid motion \cite{LeonovKM-2015-CNSNS},
and many others models (see, e.g., \cite{SprottWC-2013,WangC-2013,SharmaSPKL-2015,DangLBW-2015-HA,KuznetsovKMS-2015-HA,PhamVJWV-2014-HA,PhamJVWG-2014-HA,WeiWL-2014-HA,LiSprott-2014-HA,PhamRFF-2014-HA,WeiML-2014-HA,
PhamVVLV-2015-HA,ChenYB-2015-HA,ChenLYBXW-2015-HA,WeiZWY-2015-HA,BurkinK-2014-HA,WeiZ-2014-HA,LiZY-2014-HA,ZhaoLD-2014-HA,LaoSJS-2014-HA,ChaudhuriP-2014-HA,
PhamVJW-2014,PhamVJWW-2014,KingniJSW-2014,LiSprott-2014-PLA-cu,MolaieJSG-2013,JafariSPGJ-2014,JafariS-2013-cu}).

Similar to autonomous systems,
when analyzing and visualizing chaotic behaviors of nonautonomous systems,
we can consider the extended phase space
and introduce various notions of attractors
(see, e.g., \cite{ChebanKS-2002,KloedenM-2011}).
Alternatively, we can regard time $t$ as a phase space variable that obeys the equation $\dot t=1$.
For systems that are periodic in time, we can also introduce a cylindrical phase space and consider
the behavior of trajectories on a Poincare section.

The consideration of system equilibria and the notions of
self-excited and hidden attractors
are natural for autonomous systems,
because their equilibria can be easily found
analytically or numerically.
However, we may use other objects to construct transient processes that
lead to the discovery of chaotic sets.
These objects can be constructed for the considered system or its modifications
(i.e., instead of analyzing the scenario of the system
transiting into chaos,
we can synthesis a new transition scenario).
For example, we can use perpetual points \cite{Prasad-2015}
or the equilibria of the complexified system \cite{PhamJVWG-2014-HA}.
A periodic solution or homoclinic trajectory can be used in a similar way
(some examples of theoretical studies can be found in
\cite{CartwrightL-1945, Levinson-1949, Melnikov-1963, Shilnikov-1965};
however the presence of chaotic behavior in the considered examples
may not imply the existence of a chaotic attractor,
which can be numerically visualized using the standard computational procedure).

For nonautonomous systems, depending on the physical problem statement,
the notion of self-excited and hidden attractors can be introduced
with respect to either the stationary states ($x(t)\equiv x_0\ \forall t$)
of the considered  nonautonomous system,
the stationary points of the system at fixed initial time $t=t_0$,
or the corresponding system without time-varying excitations.
If the discrete dynamics of the system are considered on a Poincare section,
then we can also use stationary or periodic points on the section that corresponds
to a periodic orbit of the system (the consideration of periodic orbits is also natural
for discrete systems).

In the following, we consider an example of a nonautonomous system
(a forced Duffing oscillator),
so that we can visualize the chaotic behavior.
The classical example of a self-excited chaotic attractor (Fig.~\ref{fig:uead})
in a Duffing system $\ddot x+0.05\dot x+x^3=7.5\cos(t)$
was numerically constructed by Ueda in 1961,
but it become well-known much later \cite{UedaAH-1973}.
To construct a self-excited chaotic attractor in this system,
we use a transient process from the zero equilibrium
of the unperturbed autonomous system (i.e., without $\cos(t)$)
to the attractor (Fig.~\ref{fig:uead}) in the forced system.

\begin{figure}[!ht]
 \centering
 \includegraphics[width=0.5\textwidth]{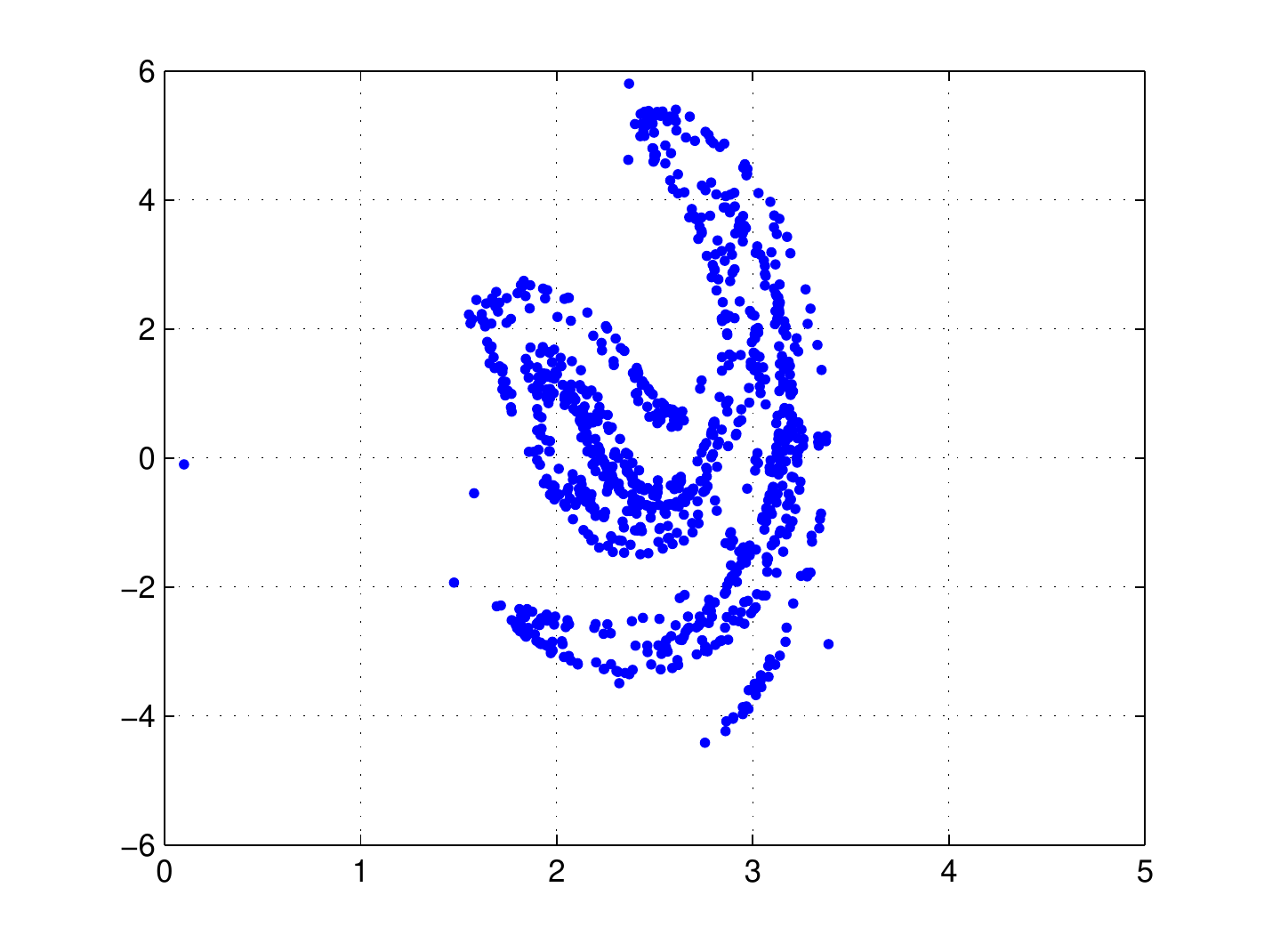}
 \caption{
   Forced Duffing oscillator: $\ddot x+0.05\dot x+x^3=7.5\cos(t)$.
   The $(x,\dot x)$ plane is mapped into itself
   by following the trajectory for time $0 \leq t \leq 2\pi$.
   After a transition process a trajectory from the vicinity of the
   zero stationary point of the
   unperturbed Duffing oscillator (without $7.5\cos(t)$) visualizes
   a self-excited chaotic attractor in the forced oscillator.
 }
 \label{fig:uead}
\end{figure}

Note that if the attracting domain is the
whole state space, then the attractor can be visualized by
any trajectory and the only difference between computations
is the timing of the transient process.

\section{A Lorenz-like system}
\label{intro:lorenz}

Consider a three-dimensional Lorenz-like system
\begin{equation}
\begin{cases}
 \dot{x} $ = $ - \sigma(x-y) - ayz\\
 \dot{y} $ = $ rx-y-xz \\
 \dot{z} $ = $ -bz+xy
\end{cases}\label{sys:lorenz-general}
\end{equation}
For  $a = 0$, system \eqref{sys:lorenz-general} coincides
with  the classical Lorenz system \cite{Lorenz-1963}.
For $\sigma > ar$ and $b = 1$ after a linear change of variables \cite{LeonovB-1992}
\begin{equation}
x \rightarrow x, \quad y \rightarrow \frac{C}{\sigma - ar} z, \quad
z \rightarrow r - \frac{C}{\sigma - ar} y
\label{sys:lorenz-general:change_var}
\end{equation}
system~\eqref{sys:lorenz-general} takes the following form
\begin{equation}
\begin{cases}
 \dot{x} $ = $ -\sigma x + Cz + A yz \\
 \dot{y} $ = $  R_a - y - xz \\
 \dot{z} $ = $ -z + xy
\end{cases}
\label{sys:conv_fluid}
\end{equation}
with
\begin{equation}
 C > 0, \quad
 R_a = \frac{r (\sigma - ar)}{C} > 0, \quad
 A = \frac{C^2 a}{(\sigma - ar)^2} > 0.
 \label{sys:conv_fluid:param}
\end{equation}
System \eqref{sys:conv_fluid} was suggested by Glukhovsky and
Dolghansky~\cite{GlukhovskyD-1980}, and describes convective
fluid motion in an ellipsoidal rotating cavity,
which can be interpreted as one of the models of ocean flows
(see Appendix~\ref{appendix:phys-problem} for a description of this problem).

In \cite{LeonovB-1992}, system \eqref{sys:lorenz-general} 
was obtained as a linear transformation of the Rabinovich system \cite{Rabinovich-1978}.
It describes interactions between waves in plasma
\cite{Rabinovich-1978,PikovskiRT-1978}. 
Additionally, system \eqref{sys:lorenz-general}
describes the following physical processes \cite{LeonovB-1992}:
a rigid body rotation in a resisting medium,
the forced motion of a gyrostat,
a convective motion in harmonically oscillating horizontal fluid layer,
and Kolmogorov flow.
Systems \eqref{sys:lorenz-general} and \eqref{sys:conv_fluid}
are interesting because of the discovery of chaotic attractors
in their phase spaces.
Moreover, system \eqref{sys:lorenz-general} was used to describe
the specific mechanism of transition to chaos
in low-dimensional dynamical systems (gluing bifurcations) \cite{AkhtanovZZ-2013}.

For system \eqref{sys:lorenz-general} with $\sigma = \pm a r$, \cite{EvtimovPS-2000}
contains a detailed analysis of
equilibria stability and the asymptotic behavior of trajectories, and a derivation of
the parameter values for which the system is integrable.
Other researchers have also considered the analytical and numerical analysis of some
extensions of system \eqref{sys:lorenz-general} \cite{PanchevSV-2007,LiaoTA-2010}.

Further, following \cite{GlukhovskyD-1980}, we consider system \eqref{sys:lorenz-general}
with
\[
  b=1, \quad a>0, \quad r>0, \quad \sigma>ar.
\]

\subsection{Classical scenario of the transition to chaos}

For the Lorenz system \cite{Sparrow-1982}, 
the following classical scenario of transition to chaos is known.
Suppose that $\sigma$ and $b$ are fixed
(let us consider the classical parameters $\sigma = 10$, $b = 8/3$),
and that $r$ varies.
Then, as $r$ increases,
the phase space of the Lorenz system is subject to the following sequence of bifurcations.
For $0<r<1$, there is globally asymptotically stable zero equilibrium $S_0$.
For $r>1$, equilibrium $S_0$ is a saddle,
and a pair of symmetric equilibria $S_{1,2}$ appears.
For  $1 < r < r_{h} \approx 13.9$, the separatrices $\Gamma_{1,2}$ of equilibria $S_0$
are attracted to the equilibria $S_{1,2}$.
For $r = r_{h} \approx 13.9$, the separatrices $\Gamma_{1,2}$
form two homoclinic trajectories of equilibria $S_0$  (homoclinic butterfly).
For  $r_h < r < r_c \approx 24.06$,
the separatrices $\Gamma_1$ and $\Gamma_2$ tend to $S_2$ and $S_1$, respectively.
For  $r_c < r$, the separatrices $\Gamma_{1,2}$ are attracted to a self-excited attractor
(see, e.g., \cite{Sparrow-1982,YuC-2004}).
For $r>r_a$, the equilibria $S_{1,2}$ become unstable.
Finally, $r=28$ corresponds to the classical self-excited Lorenz attractor
(see Fig.~\eqref{fig:lorenz:attr:se}).


Furthermore, it has been shown that system \eqref{sys:lorenz-general} follows
a similar scenario of transition to chaos.
However, a substantial distinction of
this scenario is the presence of hidden chaotic attractor
in the phase space of system \eqref{sys:lorenz-general}
for certain parameters values \cite{LeonovKM-2015-CNSNS}.

Let us determine the stationary points of system \eqref{sys:lorenz-general}.
We can show that for positive parameters,
if $r < 1$, system \eqref{sys:lorenz-general} has a unique equilibrium
${\bf \rm S_0} = (0,0,0)$,
which is globally asymptotically Lyapunov stable \cite{BoichenkoLR-2005}.
If $r > 1$, then \eqref{sys:lorenz-general} possesses three equilibria:
a saddle ${\bf \rm S_0} = (0,0,0)$ and symmetric (with respect to $z = 0$) equilibria
\begin{equation}
 {\bf \rm S_{1,2}} = (\pm x_1, \pm y_1, z_1), \label{eq:equil_s12}
\end{equation}
where
$$
 x_1 = \frac{\sigma \sqrt{\xi}}{\sigma + a \xi}, \quad
 y_1 = \sqrt{\xi}, \quad
 z_1 = \frac{\sigma \xi}{\sigma 1 + a \xi},
$$
and 
$$
 \xi = \frac{\sigma }{2 a^2} \left[ a (r-2) -
 \sigma + \sqrt{(\sigma - ar)^2 + 4a\sigma} \right].
$$
The characteristic polynomial of the Jacobian matrix of system \eqref{sys:lorenz-general}
at the point $(x,\,y,\,z)$ has the form
\[
	\chi(x,\,y,\,z) = \lambda^3 +  p_1(x,\,y,\,z) \lambda^2 +
	p_2(x,\,y,\,z) \lambda + p_3(x,\,y,\,z),
\]
where
\begin{align*}
& p_1(x,\,y,\,z) =  \sigma + 2, \quad
p_2(x,\,y,\,z) = x^2 + a y^2 - a z^2 +
	(\sigma + a r)z - r \sigma + 2 \sigma + 1, \\
& p_3(x,\,y,\,z) = \sigma x^2 + a y^2 -
a z^2 - 2 a x y z + (\sigma  + a r) x y +
(\sigma  + a r) z - r \sigma + \sigma.
\end{align*}
Following \cite{GlukhovskyD-1980}, we let $\sigma = 4$
and define the stability domain of equilibria $S_{1,2}$.
Using the Routh-Hurwitz criterion we can obtain
the following (see Appendix~\ref{appendix:stability}).
\begin{proposition}\label{proposition:gen_lorenz:stability}
The equilibria $S_{1,2}$ are stable if
\begin{equation}\label{eq:ineq:rA2}
 8 a^2 r^3 + a(7 a - 64) r^2 + (288 a + 128) r + 256 a - 2048 < 0.
\end{equation}
\end{proposition}

\begin{figure}[ht!]
 \centering
 \includegraphics[width=0.55\textwidth]{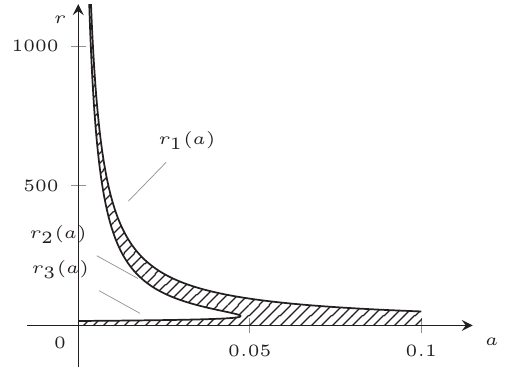}
 \caption{The stability domain of equilibria $S_{1,2}$
  for  $\sigma = 4$.}
 \label{fig:gen_lorenz:regimes}
\end{figure}

The discriminant of the left-hand side of \eqref{eq:ineq:rA2} has only
one positive real root, $a^* \approx 0.04735$.
So the roots of the polynomial in \eqref{eq:ineq:rA2} are as follows.
For $0 < a < a^*$, there are three real roots $r_1 (a) > r_2 (a) > r_3 (a)$;
for $a = a^*$, there are two real roots: $r_1 (a)$ and $r_2(a) = r_3(a)$;
for  $a > a^*$, there is one real root $r_1 (a)$.
Thus, for  $0 < a < a^*$, the equilibria $S_{1,2}$ are stable
for  $r < r_3 (a)$ and for  $r_2 (a) < r < r_1 (a)$;
and for  $a > a^*$ the equilibria $S_{1,2}$ are stable for $r < r_1 (a)$
(see Fig. \ref{fig:gen_lorenz:regimes}).

\begin{figure}[!ht]
 \centering
 \includegraphics[width=0.7\textwidth]{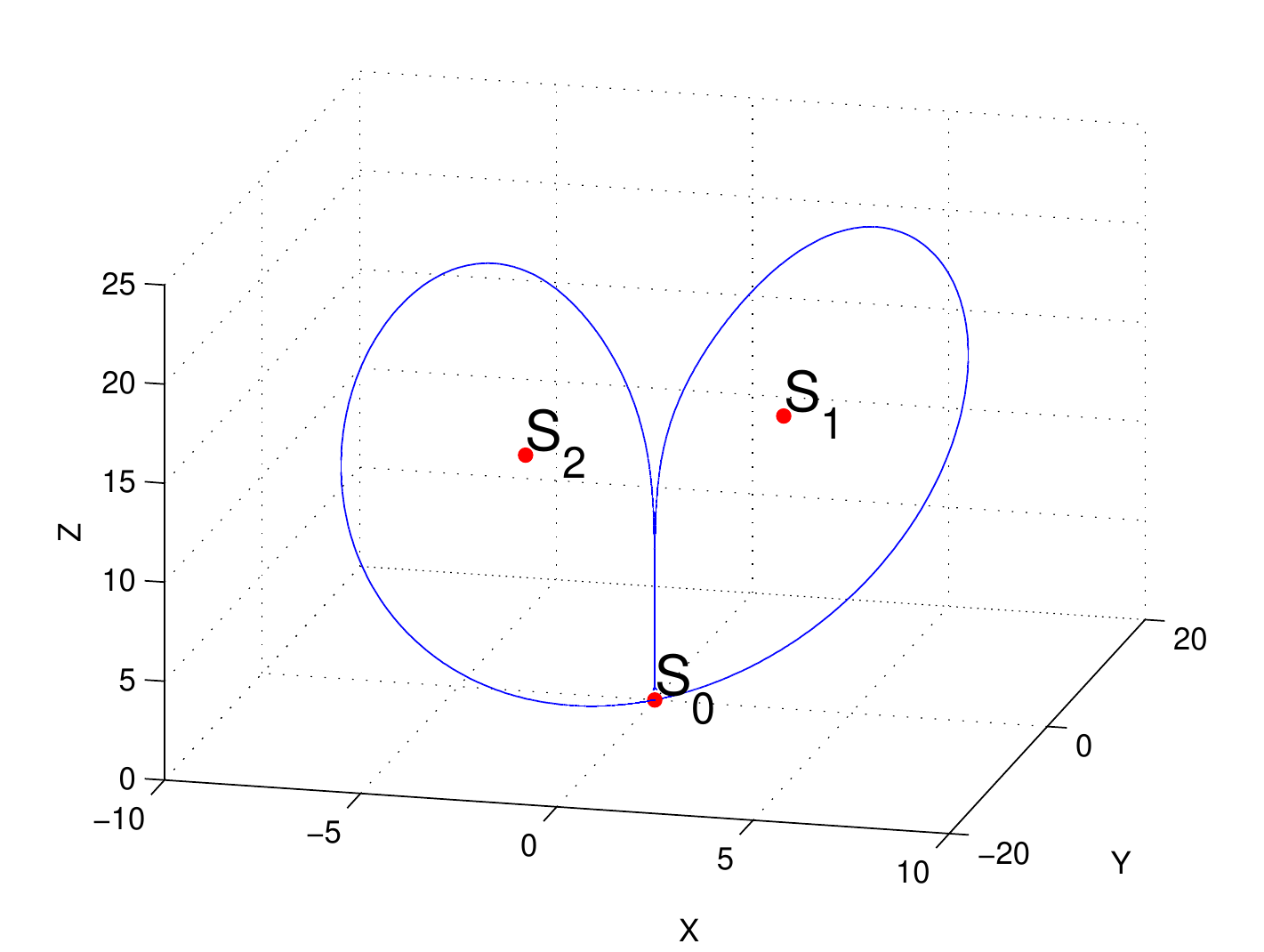}
 \caption{Separatrices of the saddle $S_0 = (0,0,0)$ of system
 \eqref{sys:lorenz-general} for  $\sigma = 4, \, a = 0.0052...$, and $r \approx 7.443$.}
 \label{fig:separatrix}
\end{figure}

Consider the problem of the existence of a homoclinic orbit,
which is important in bifurcation theory
and in scenarios of transition to chaos (see, e.g., \cite{AfraimovichGLShT-2014}).
For \eqref{sys:lorenz-general} and \eqref{sys:conv_fluid},
we can prove the existence of homoclinic trajectories for the zero saddle equilibrium $S_0$
using the fishing principle
\cite{Leonov-2012-PLA,Leonov-2013-IJBC,Leonov-2014-ND,LeonovK-2015-AMC}.
The fishing principle is based on the construction of
a special two-dimensional manifold such that a separatrix of a saddle point intersects
or does not intersect the manifold for two different system parameter values.
Continuity implies the existence of some intermediate parameter value for which the
separatrix touches the manifold.
According to the construction, the separatrix must
touch a saddle point, so we can numerically localize the birth of a homoclinic orbit.
A rigorous description is given in Appendix~\ref{appendix:homo}.

For $\sigma = 4, a = 0.0052$, and  $r \approx 7.443$
we numerically obtain a homoclinic trajectory (see Fig.~\ref{fig:separatrix}).

We come now to the study of the limit behaviors of trajectories and attractors.
We introduce some rigorous notions of a dynamical system and attractor and discuss
the connection with the notions of self-excited
and hidden attractors from a computational perspective.

\section{Definitions of attractors}

\subsection{Dynamical systems and ordinary differential equations}

Consider an autonomous system of the differential equations
\begin{equation}
 \dot{{\bf \rm x}} = {\bf \rm f}({\bf \rm x}), \quad
 t \in \mathbb{R},~ {\bf \rm x} \in \mathbb{R}^n,
 \label{eq:ode}
\end{equation}
where ${\bf \rm f} : \mathbb{R}^n \to \mathbb{R}^n$ is a continuous
vector-function that satisfies {\it a local Lipschitz condition} in $\mathbb{R}^n$.
The Picard theorem (see, e.g., \cite{CoddingtonL-1955,Hartman-1964})
for a local Lipschitz condition on the function ${\bf \rm f}$ implies that,
for any ${\bf \rm x}_0 \in \mathbb{R}^n$, there exists a unique solution
${\bf \rm x}(t,{\bf \rm x}_0)$ to differential equation \eqref{eq:ode}
with the initial data ${\bf \rm x}(t_0,{\bf \rm x}_0) = {\bf \rm x}_0$,
which is given on a certain finite time interval: $t \in I \subset \mathbb{R}$.
The theorem regarding the continuous dependence on initial data
\cite{CoddingtonL-1955,Hartman-1964}\footnote{
Similar theorems on the existence, unicity,
and  continuous dependence on the initial data
for solutions of system with the discontinuous right-hand side
are considered in \cite{YakubovichLG-2004,KiselevaK-2015-Vestnik}.}
implies that the solution ${\bf \rm x}(t,{\bf \rm x}_0)$
continuously depends on ${\bf \rm x}_0$.

To study the limit behavior of trajectories and compute the limit values,
characterizing trajectories, we consider the solutions of \eqref{eq:ode} for
$t \to +\infty$ or $t \to \pm \infty$.
For arbitrary quadratic systems, the existence of solutions
for  $t \in [t_0,\, +\infty)$ does not generally imply the existence of solutions for  $t \in (-\infty,\, t_0]$
(see the classical one-dimensional example $\dot{x} = x^2$  or multidimensional
examples from the work on the completeness
of quadratic polynomial systems \cite{GingoldS-2011}).
It is known that if ${\bf \rm f}$ is continuously
differentiable (${\bf \rm f} \in C^1$), then ${\bf \rm f}$
is locally Lipschitz continuous in $\mathbb{R}^n$ (see, e.g., \cite{HirschSD-2004}).
Additionally, if ${\bf \rm f}$ is locally Lipschitz continuous,
 then for any ${\bf \rm x}_0 \in \mathbb{R}^n$
 the solution ${\bf \rm x}(\cdot, \, {\bf \rm x}_0):I \to \mathbb{R}^n$
 exists on maximal time interval $I = (t_{-}, \, t_{+}) \in \mathbb{R}$,
 where $-\infty \leq t_{-} < t_{+} \leq + \infty$.
 If $t_{+} < +\infty$, then $||{\bf \rm x}(t, {\bf \rm x}_0)|| \to \infty$
 for  $t \to t_{+}$, and if
 $t_{-} > -\infty$, then $||{\bf \rm x}(t, {\bf \rm x}_0)|| \to \infty$
 for  $t \to t_{-}$ (see, e.g., \cite{Teschl-2012}).
This implies that a solution of \eqref{eq:ode}
is continuous if it remains bounded.
For convenience, we introduce a set of time values
$\mathbb{T} \in \left\{\mathbb{R},\mathbb{R}_{+}\right\}$.
The existence and uniqueness of solutions of \eqref{sys:lorenz-general} for all $t \in \mathbb{T}$
can be provided, for example, by a {\it global Lipschitz condition}.
Another effective method for studying the boundedness
of solutions for all $t \in \mathbb{T}$ is to construct a Lyapunov function.

If the existence and uniqueness conditions for all $t \in \mathbb{T}$
are satisfied, then: 1) the solution of \eqref{eq:ode}
satisfies the group property (\cite{CoddingtonL-1955,Hartman-1964})
\begin{equation}\label{group_prop}
 {\bf \rm x}(t+s, {\bf \rm x}_0) = {\bf \rm x}(t,{\bf \rm x}(s,{\bf \rm x}_0)),
 \quad \forall ~ t, s \in \mathbb{T},
\end{equation}
and 2)
${\bf \rm x}(\cdot, \cdot) : \mathbb{T} \times
\mathbb{R}^n \to \mathbb{R}^n$ is a continuous mapping
according to the theorem of the continuous dependence of the
solution on the initial data.
Thus, if the solutions of \eqref{eq:ode} exist and satisfy
\eqref{group_prop} for all $t \in \mathbb{T}$,
the system generates a \emph{dynamical system} \cite{Birkhoff-1927} on the phase space
$(\mathbb{R}^n, ||\cdot||)$.
Here $||{\bf \rm x}|| = \sqrt{x_1^2 + \cdots + x_n^2}$ is
an Euclidean norm of the vector
${\bf \rm x} = (x_1, \ldots, x_n) \in \mathbb{R}^n$, which generates a metric on
$\mathbb{R}^n$.
We abbreviate {\it ``dynamical system generated by
a differential equation''} \, to {\it ``dynamical system''}.
Because the initial time is not important for dynamical systems,
without loss of generality we consider
\[
   x(t,x_0):\ x(0,x_0)=x_0.
\]

Consider system \eqref{sys:lorenz-general}.
Its right-hand side is continuously differentiable
in $\mathbb{R}^n$, which means that it is locally Lipschitz continuous in $\mathbb{R}^n$
(but not globally Lipschitz continuous).
Analogous with the results for the Lorenz system \cite{Meisters-1989,Coomes-1989},
we can prove that the
solutions of \eqref{sys:lorenz-general} exist for all $t \in \mathbb{R}$,
i.e. system \eqref{sys:lorenz-general} is invertible.
For this purpose, we can use the Lyapunov function (Appendix~\ref{appendix:dissip})
\begin{equation}\label{flgd}
 V(x,y,z) = \frac{1}{2} \left( x^2 + y^2 + (a+1)
 \left(z - \frac{\sigma + r}{a + 1} \right)^2 \right) \geq 0.
\end{equation}
Then, system \eqref{sys:lorenz-general} generates a dynamical system
and we can study its limit behavior and attractors.

\subsection{Classical definitions of attractors}
The notion of an attractor is connected with investigations of the limit
behavior of the trajectories of dynamical systems.
We define attractors as follows
\cite{Ladyzhenskaya-1987,Ladyzhenskaya-1991,BabinV-1992,Chueshov-1993,Temam-1997,Chueshov-2002-book,BoichenkoLR-2005,Leonov-2008}.

\begin{definition}\label{def:inv_set}
  A set $K$ is said to be \emph{positively invariant} for a dynamical system if
  $$
    {\bf \rm x}(t, K) \subset K, \ \forall t \geq 0,
  $$
 and to be \emph{invariant} if
  $$
    {\bf \rm x}(t, K) = K, \ \forall t \geq 0,
  $$
 where ${\bf \rm x}(t, K) = \left\{ {\bf \rm x}(t, {\bf \rm x}_0) ~ | ~
 {\bf \rm x}_0 \in K,\ t \geq 0 \right \}$.
\end{definition}

\begin{property}\label{property:local_attr_set}
Invariant set $K$ is said to be {\it locally attractive} for a dynamical system
if, for a certain $\varepsilon$-neighborhood $K(\varepsilon)$ of set $K$,
$$
 \lim_{t \to +\infty} \rho (K, {\bf \rm x}(t, {\bf \rm x}_0)) = 0, \quad \forall ~
 {\bf \rm x}_0 \in K(\varepsilon).
$$
Here $\rho(K, {\bf \rm x})$ is a distance from the point ${\bf \rm x}$
to the set $K$, defined as
$$
 \rho(K, {\bf \rm x}) = \inf_{{\bf \rm z} \in K} || {\bf \rm z} - {\bf \rm x}
||,
$$
and $K(\varepsilon)$ is a set of points ${\bf \rm x}$ for which
$\rho (K, {\bf \rm x}) < \varepsilon$.
\end{property}

\begin{property}\label{property:global_attr_set}
Invariant set $K$ is said to be
{\it globally attractive} for dynamical system if
$$
 \lim_{t \to +\infty} \rho (K, {\bf \rm x}(t, {\bf \rm x}_0)) = 0, \quad \forall ~
 {\bf \rm x}_0 \in \mathbb{R}^n.
$$
\end{property}

\begin{property}\label{property:unif_loc_attr_set}
Invariant set $K$ is said to be
{\it uniformly locally attractive} for a dynamical
system if for a certain $\varepsilon$-neighborhood
$K(\varepsilon)$,
any number $\delta > 0$, and any bounded set $B$,
there exists a number $t(\delta, B) > 0$ such that
$$
 {\bf \rm x}(t, B \cap K(\varepsilon)) \subset K(\delta),
 \quad \forall ~ t \geq t(\delta, B).
$$
Here
$$
 {\bf \rm x}(t, B \cap K(\varepsilon)) = \left\{{\bf \rm x}
 (t, {\bf \rm x}_0) ~|~ {\bf \rm x}_0 \in B \cap
 K(\varepsilon) \right\}.
$$
\end{property}

\begin{property}\label{property:unif_glob_attr_set}
Invariant set $K$ is said to be
{\it uniformly globally attractive} for a dynamical
 system if, for any number $\delta > 0$
and any bounded set $B \subset \mathbb{R}^n$,
there exists a number $t(\delta, B) > 0$ such that
$$
 {\bf \rm x}(t, B) \subset K(\delta), \quad \forall ~ t \geq t(\delta, B).
$$
\end{property}

\begin{definition}\label{def:attractor}
For a dynamical system, a bounded closed
invariant set $K$ is:
\begin{enumerate}[label=(\arabic*)]
 \item an \emph{attractor} if it is a locally attractive set
       (i.e., it satisfies Property~\ref{property:local_attr_set});
 \item a \emph{global attractor} if it is a globally attractive set
       (i.e., it satisfies Property~\ref{property:global_attr_set});
 \item a \emph{B-attractor} if it is a uniformly locally attractive set
       (i.e., it satisfies Property~\ref{property:unif_loc_attr_set}); or
 \item a \emph{global B-attractor} if it is a uniformly globally attractive set
       (i.e., it satisfies Property~\ref{property:unif_glob_attr_set}).
\end{enumerate}
\end{definition}

\begin{remark}\label{remark:minimal-attr}
In the definition of an attractor we assume closeness for the sake of uniqueness.
This is because the closure of a locally attractive invariant set $K$ is also a
locally attractive invariant set
(for example, consider an attractor with excluded one of the embedded unstable periodic orbits).
The closeness property is sometimes omitted from the attractor definition (see, e.g., \cite{Babin-2006}).
Additionally, the boundedness property is sometimes omitted (see, e.g., \cite{ChepyzhovG-1992}).
For example, a global attractor in a system describing a pendulum motion
is not bounded in the phase space $\mathbb{R}^2$ (but it is bounded in the cylindrical phase space).
Unbounded attractors are considered for nonautonomous systems in the extended phase space.
Note that if a dynamical system is defined for $t \in \mathbb{R}$,
then a locally attractive invariant set only contains the whole trajectories, i.e.
if $x_0 \in K$, then $x(t,x_0) \in K$ for $\forall t \in \mathbb{R}$ (see \cite{Chueshov-2002-book}).
\end{remark}
\smallskip

\begin{remark}\label{remark:minimal-attr}
The definition considered here implies that a global B-attractor
is also a global attractor.
Consequently, it is rational to introduce the notion of a
{\it minimal global attractor} (or {\it minimal attractor})
\cite{Chueshov-1993,Chueshov-2002-book}.
This is the smallest bounded closed invariant set that
possesses Property~\ref{property:global_attr_set}
(or Property~\ref{property:local_attr_set}).
Further, the attractors (global attractors)
will be interpreted as minimal attractors (minimal global attractors).
\end{remark}
\smallskip

\begin{definition}
For an attractor $K$, the {\it basin of attraction}
is a set ${\it B} (K) \subset \mathbb{R}^n$ such that
\[
 \lim_{t \to +\infty} \rho (K, {\bf \rm x}(t, {\bf \rm x}_0)) = 0, \quad
 \forall ~ {\bf \rm x}_0 \in {\it B}(K).
\]
\end{definition}

\begin{remark}
From a computational perspective,
it is not feasible to numerically check Property~\ref{property:local_attr_set}
for all initial states of the phase space of a dynamical system.
A natural generalization of the notion of an attractor is
consideration of the weaker attraction requirements: almost everywhere or
on a set of the positive Lebesgue measure (see, e.g., \cite{Milnor-2006}).
See also \emph{trajectory attractors} \cite{Sell-1996,ChepyzhovV-2002,ChueshovS-2005}.
To distinguish an artificial computer generated chaos from
a real behavior of the system one can consider
the shadowing property of the system (see, e.g., the survey in \cite{Pilyugin-2011}).

We can typically see an attractor (or global attractor)
in numerical experiments.
The notion of a B-attractor is mostly used in the theory of dimensions,
where we consider invariant sets covered by balls.
The uniform attraction requirement
in Property~\ref{property:unif_loc_attr_set}
implies that a global B-attractor involves
a set of stationary points ($\mathcal{S}$) and the corresponding unstable manifolds
$W^u (\mathcal{S}) = \left\{{\bf \rm x}_0 \in \mathbb{R}^n ~|~
\lim_{t \to -\infty} \rho (\mathcal{S}, {\bf \rm x}(t, {\bf \rm x}_0)) = 0\right\}$
(see, e.g., \cite{Chueshov-1993,Chueshov-2002-book}).
The same is true for B-attractor
if the considered neighborhood $K(\varepsilon)$
in Property~\ref{property:unif_loc_attr_set}
contains some of the stationary points from $\mathcal{S}$.
From a computational perspective,
numerically checking Property~\ref{property:unif_loc_attr_set}
is also difficult.
Therefore if the basin of attraction involves unstable manifolds of equilibria,
then computing the minimal attractor and the
unstable manifolds that are attracted to it may be regarded
as an approximation of B-attractor.
For example, consider the visualization of the classical Lorenz attractor
from the neighborhood of the zero saddle equilibria.
Note that a minimal global attractor involves
the set $\mathcal{S}$ and
its basin of attraction involves the set $W^u (\mathcal{S})$.
Various analytical-numerical methods for computing attractors
and their basins of attraction can be found in, for example,
\cite{Zubov-1964,Aulbach-1983,Vannelli-1985,Osipenko-1999,DellnitzJ-2002,Giesl-2007}.
\end{remark}
\smallskip

\begin{figure}[!ht]
 \centering
 \subfloat[
 {\scriptsize
  Two trajectories with symmetric initial data near $S_0$}
 ] {
 \label{fig:gen_lorenz:attr:se0}
 \includegraphics[width=0.33\textwidth]{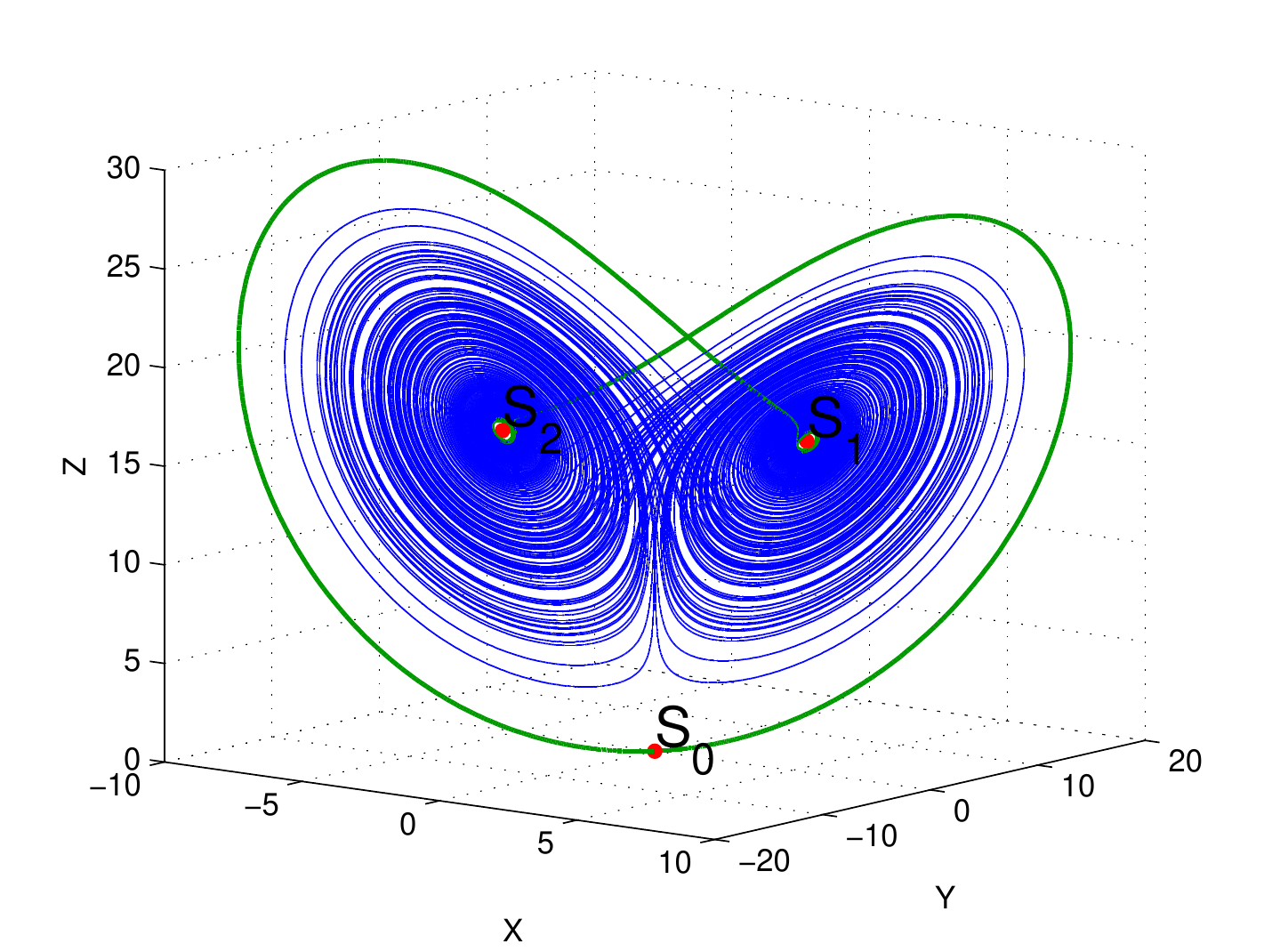}
 }~
 \subfloat[
 {\scriptsize A trajectory with initial data near $S_1$}
 ] {
 \label{fig:gen_lorenz:attr:se1}
 \includegraphics[width=0.33\textwidth]{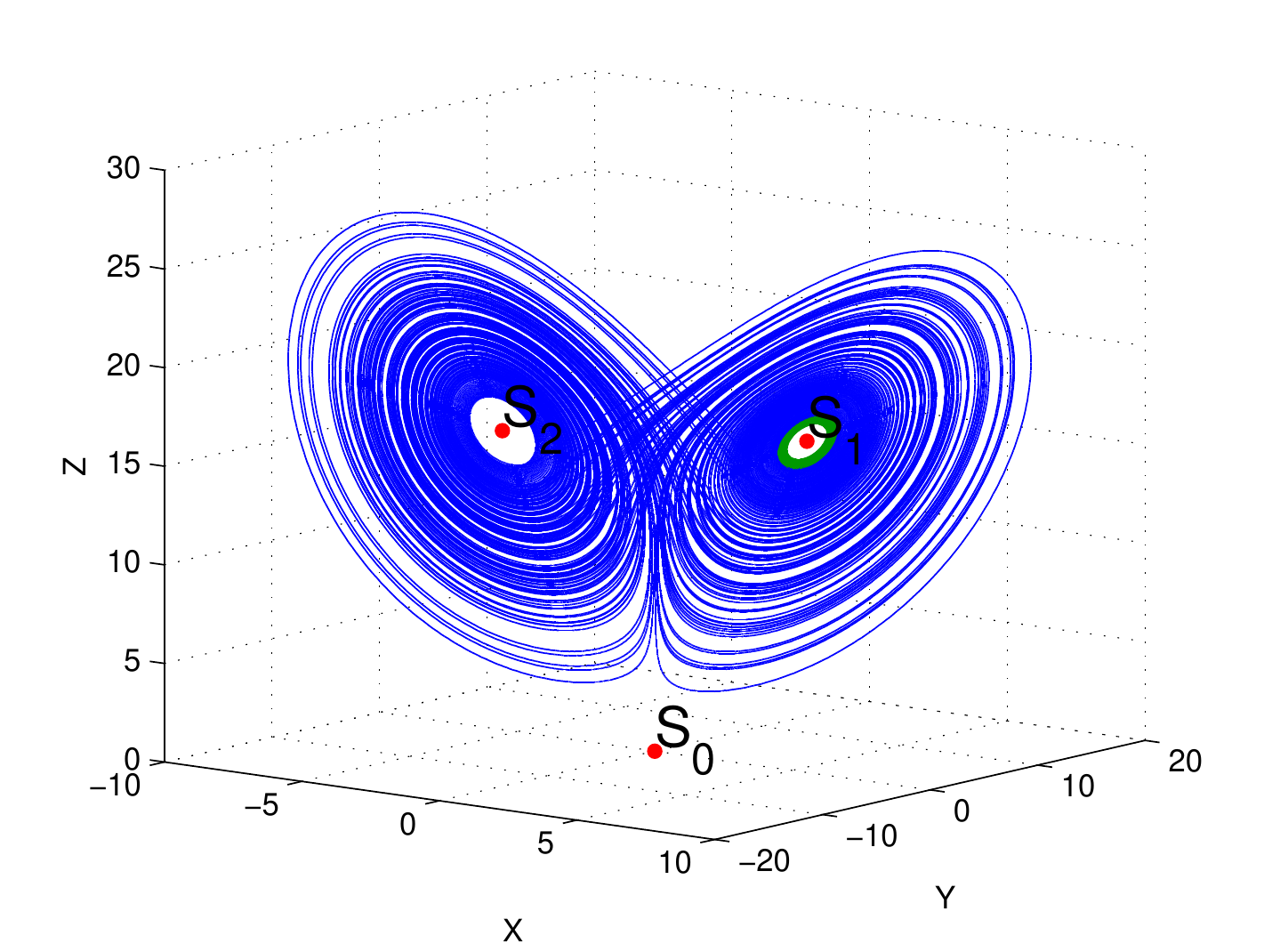}
 }
 \subfloat[
 {\scriptsize A trajectory with initial data near $S_2$}
 ] {
 \label{fig:gen_lorenz:attr:se2}
 \includegraphics[width=0.33\textwidth]{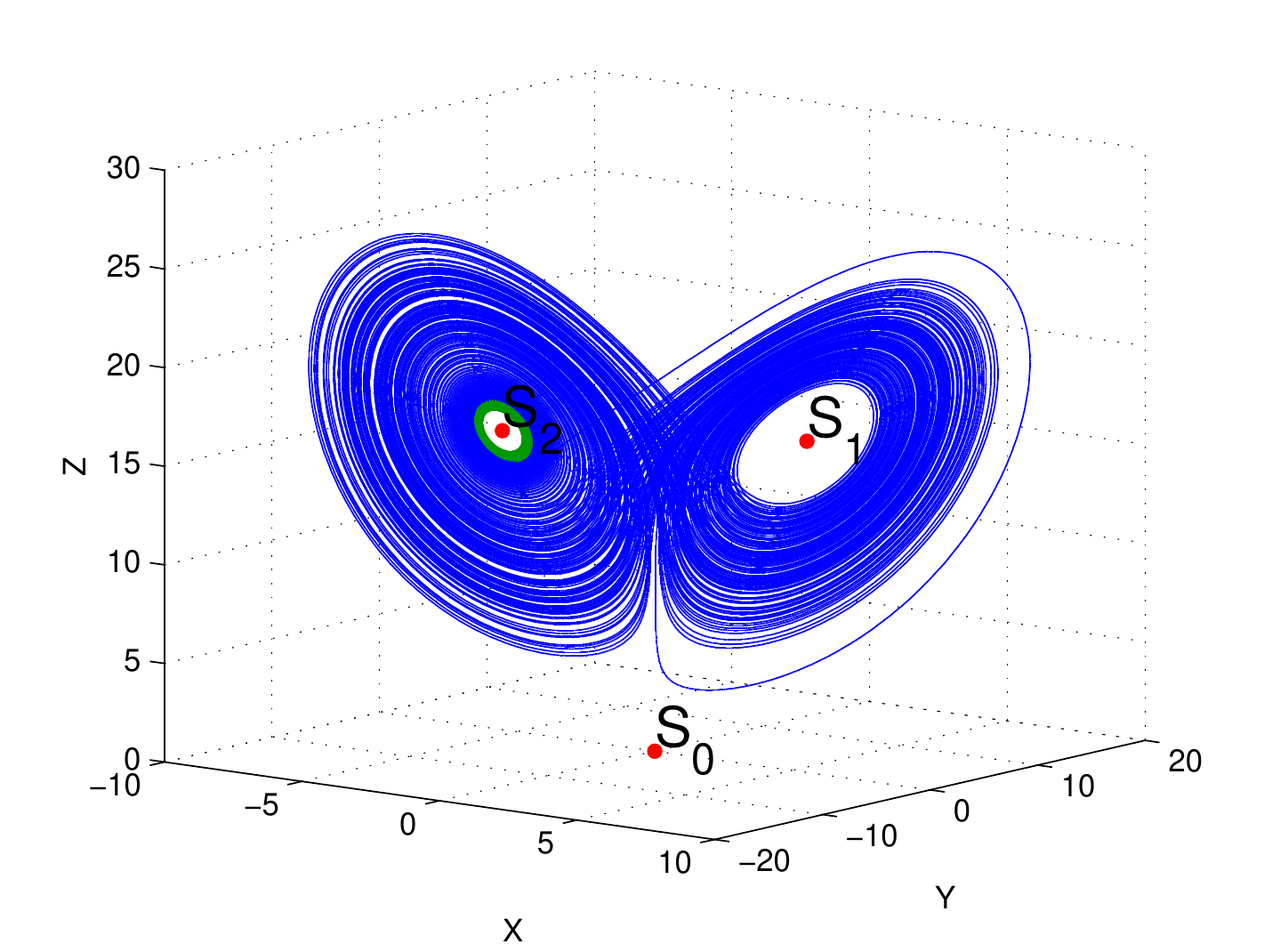}
 }
 \caption{
 Self-excited attractor of system \eqref{sys:lorenz-general} for
 $r = 17$, $\sigma = 4$, and $a = 0.0052$,
 computed from different initial points.
 } 
 \label{fig:gen_lorenz:attr:se}
\end{figure}

\section{\label{sec:self-exc:lorenz} Self-excited attractor localization}

In \cite{GlukhovskyD-1980} system \eqref{sys:conv_fluid} with $\sigma = 4$ was studied.
Consider the following parameters for system \eqref{sys:lorenz-general}
\[
 \sigma = 4, \quad a = 0.0052.
\]


According to Proposition~\ref{proposition:gen_lorenz:stability},
if $r_1\approx 16.4961242... < r < r_2 \approx 690.6735024$,
the equilibria $S_{1,2}$ of system \eqref{sys:lorenz-general}
become (unstable) saddle-focuses.
For example, if  $r = 17$, the eigenvalues of the equilibria of system
\eqref{sys:lorenz-general} are the following
\begin{eqnarray*}
   S_0 : & 5.8815, \quad -1, \quad -10.8815 & \\
   S_{1,2} : & 0.0084 \pm 4.5643 \mi, \quad -6.0168 &
\end{eqnarray*}
and there is a self-excited chaotic attractor
in the phase space of system \eqref{sys:lorenz-general}.
We can easily visualize this attractor (Fig. \ref{fig:gen_lorenz:attr:se})
using the standard computational procedure
with initial data in the vicinity of one of the equilibria $S_{0,1,2}$
on the corresponding unstable manifolds.
To improve the approximation of the attractor
one can consider its neighborhood and compute trajectories from
a grid of points in this neighborhood.


\section{Hidden attractor localization}

We need a special numerical method to localize the hidden attractor of system \eqref{sys:lorenz-general},
because the basin of attraction does not intersect the small neighborhoods of the unstable manifolds of the equilibria.
One effective method for the numerical localization of hidden attractors
is based on a \emph{homotopy} and \emph{numerical continuation}.
We construct a sequence of
similar systems such that the initial data
for numerically computing the oscillating
solution (starting oscillation) can be obtained analytically
for the first (starting) system.
For example, it is often possible to consider a starting
system with a self-excited starting oscillation.
Then we numerically track the transformation
of the starting oscillation while passing
between systems.

\begin{figure}[!ht]
 \centering
 \subfloat[$P_0:\ r = 687.5$. Attractor (blue), B-attractor (green, blue and red)] {
 \label{fig:gen_lorenz:sprtx:se1}
 \includegraphics[width=0.5\textwidth]{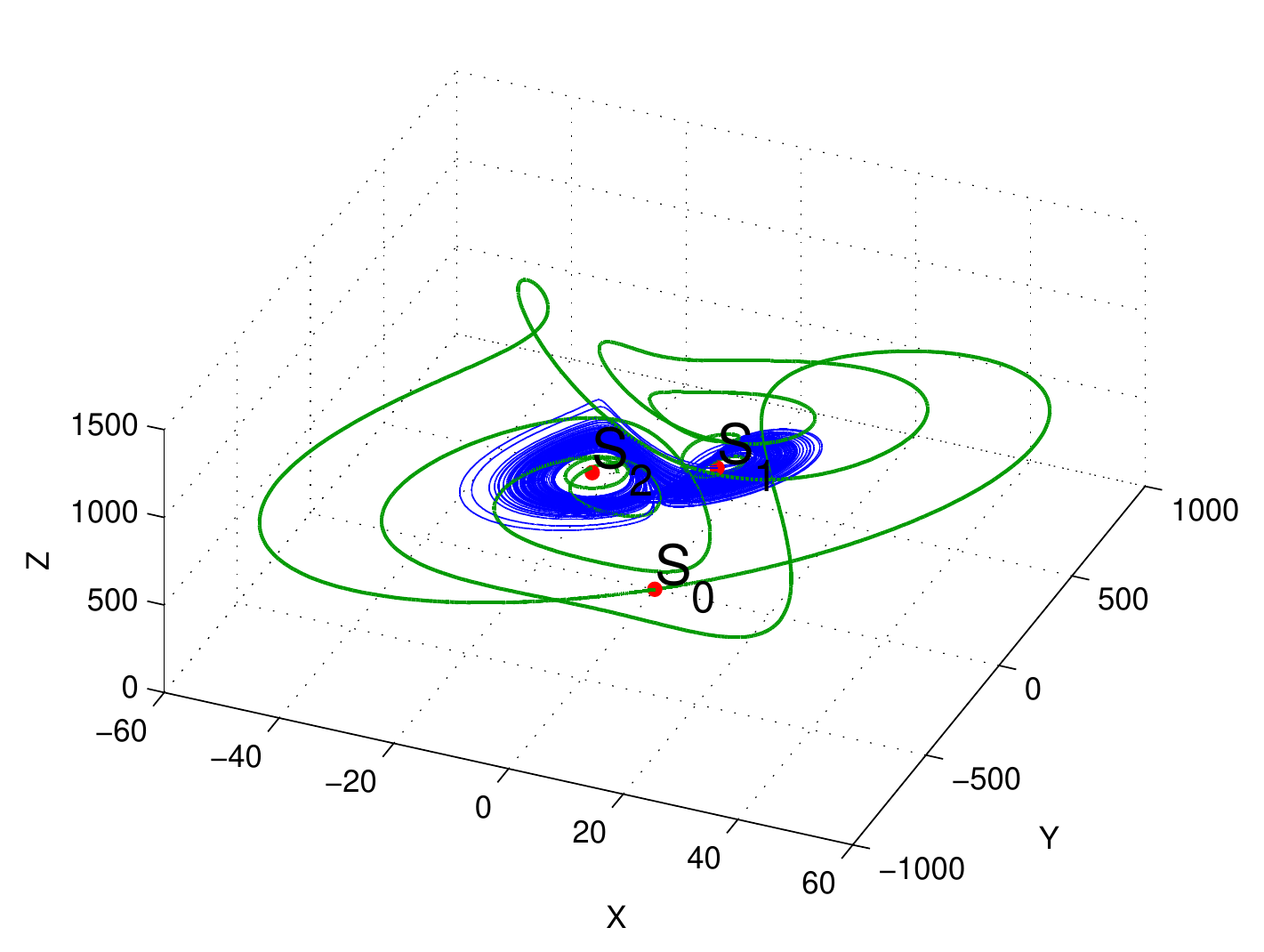}
 }
 \subfloat[$r = 690$. Attractor (blue), B-attractor (green, blue and red)] {
 \label{fig:gen_lorenz:sprtx:se2}
 \includegraphics[width=0.5\textwidth]{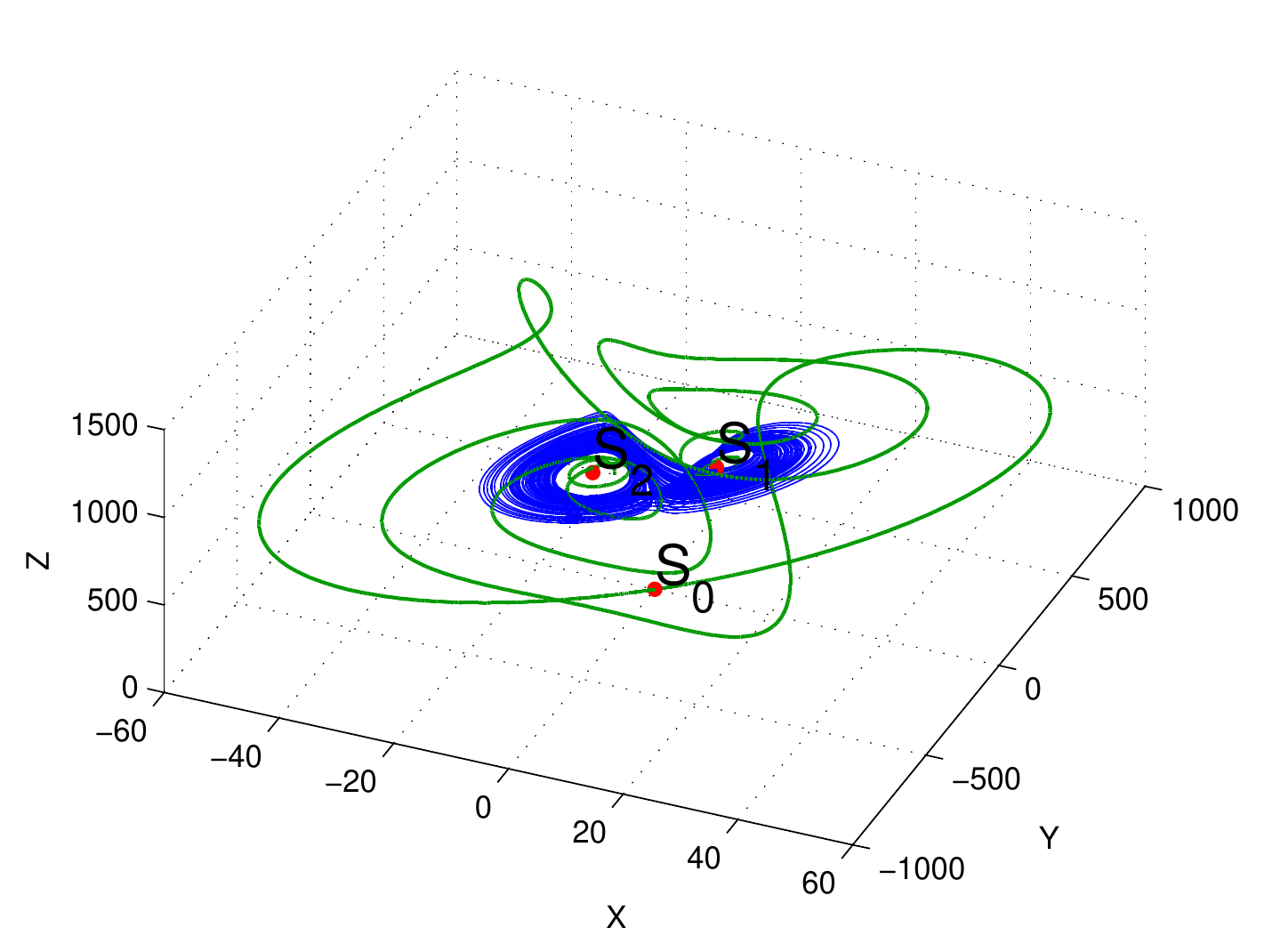}
 }

 \subfloat[$P_1:\ r = 700$. B-attractor (green and red)] {
 \label{fig:gen_lorenz:sprtx:hid}
 \includegraphics[width=0.5\textwidth]{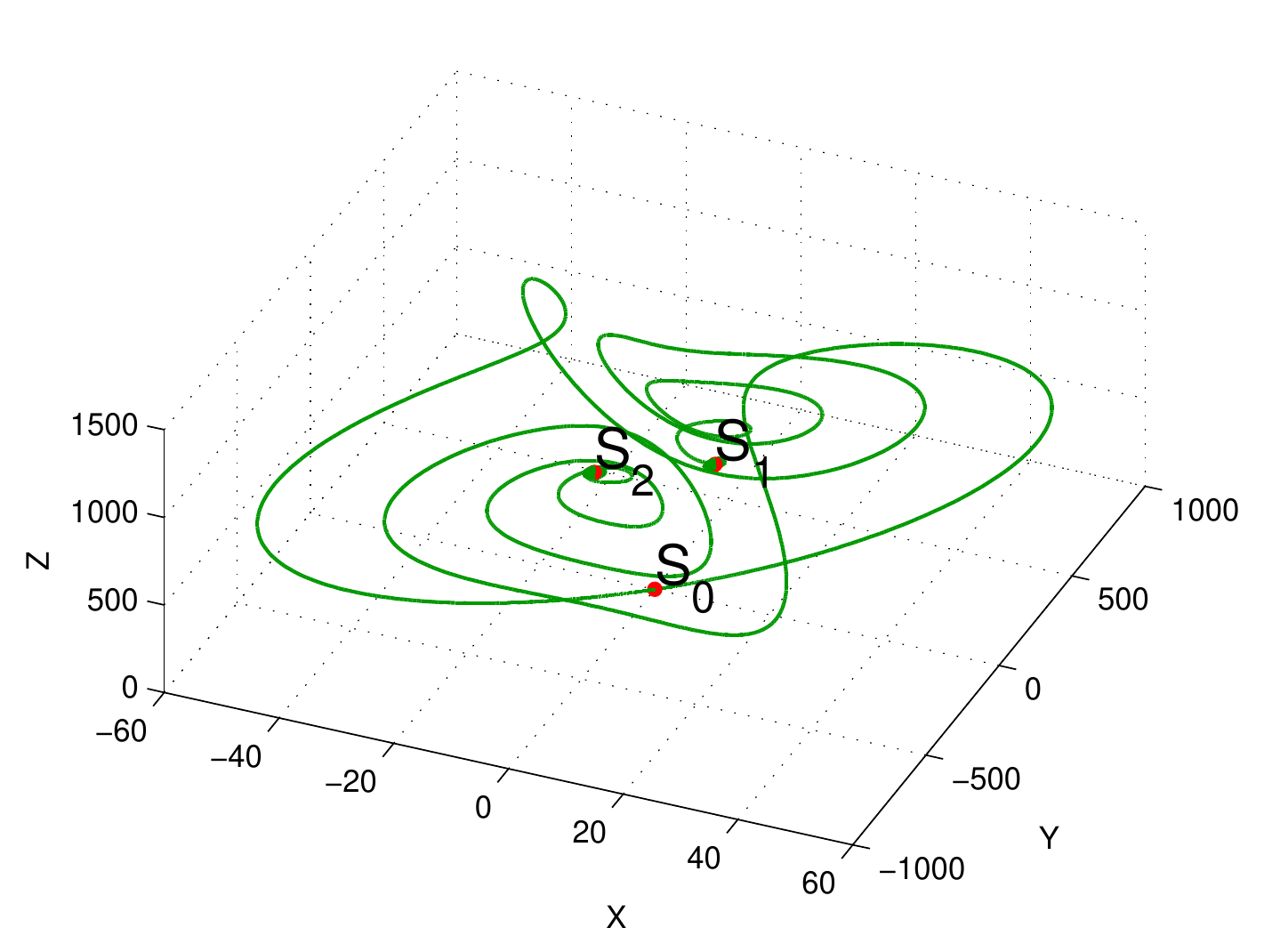}
 }
 \caption{
 B-attractor
 of system \eqref{sys:lorenz-general} for
 fixed $\sigma = 4$, $a = 0.0052$, and various $r$.
 }
\end{figure}
In a scenario of transition to chaos in dynamical system
there is typically a parameter $\lambda \in [a_1,a_2]$,
the variation of which gives the scenario.
We can also artificially introduce the parameter $\lambda$,
let it vary in the interval
$[a_1,a_2]$ (where $\lambda = a_2$ corresponds to the initial system),
and choose a parameter $a_1$ such that
we can analytically or computationally find a certain nontrivial attractor
when $\lambda = a_1$ (often this attractor has a simple form, e.g., periodic).
That is, instead of analyzing the scenario of a transition
into chaos, we can synthesize it.
Further, we consider the sequence
$\lambda_j, ~ \lambda_1 = a_1,~ \lambda_m = a_2, ~ \lambda_j \in [a_1,a_2]$
such that the distance between $\lambda_j$ and $\lambda_{j+1}$
is sufficiently small.
 Then we numerically investigate changes to the shape of the attractor
obtained for $\lambda_1 = a_1$.
 If the change in $\lambda$ (from $\lambda_j$ to $\lambda_{j+1}$) does not
 cause a loss of the stability bifurcation of the considered attractor,
 then the attractor for $\lambda_m = a_2$ (at the end of procedure)
 is localized.

Let us construct a line segment on the plane $(a,r)$ that intersects
a boundary of the stability domain of the equilibria $S_{1,2}$
(see Fig. \ref{fig:cont_proc:paths}).
We choose the point $P_1(r = 700, a = 0.0052)$ as the end point of the line segment.
The eigenvalues for the equilibria of system~\eqref{sys:lorenz-general}
that correspond to the parameters $P_1$ are the following:
\begin{eqnarray*}
S_0 : & 50.4741, \quad -1, \quad -55.4741, \\
S_{1,2} : & -0.1087 \pm 10.4543 \mi, \quad -5.7826.
\end{eqnarray*}
This means that the equilibria $S_{1,2}$ become stable focus-nodes.
Now we choose the point $P_0(r = 687.5, a = 0.0052)$
as the initial point of the line segment.
This point corresponds to the parameters for which system
\eqref{sys:lorenz-general} has a self-excited attractor, which
can be computed using the standard computational procedure.
Then we choose a sufficiently small partition step
for the line segment and compute a chaotic attractor in the phase space of system
\eqref{sys:lorenz-general} at each iteration of the procedure.
The last computed point at each step is used as the initial point
for the computation at the next step (the computation time must be sufficiently large).

\begin{figure}[!ht]
 \centering
 \includegraphics[width=0.55\textwidth]{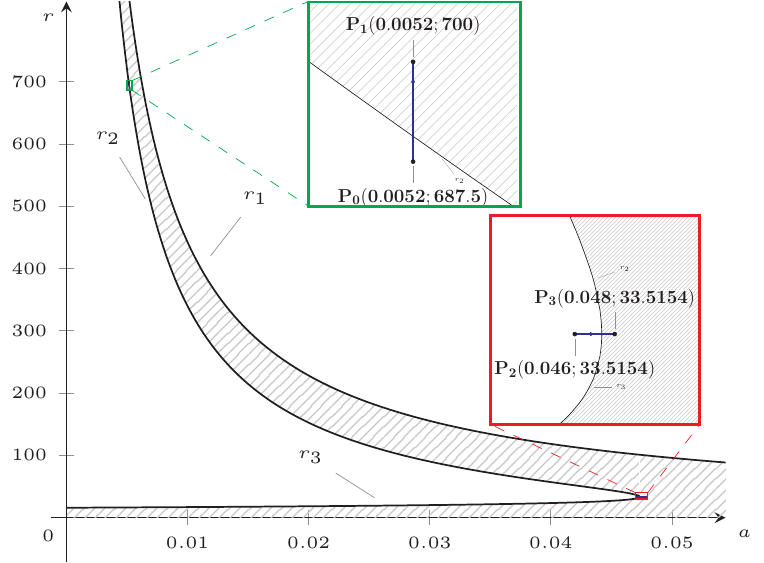}
 \caption{Paths $[P_0,P_1]$ and $[P_2,P_3]$ in the plane of parameters $\{a,r\}$
 used in the continuation procedure.}
 \label{fig:cont_proc:paths}
\end{figure}

In our experiment the length of the path was $2.5$ and there were $6$ iterations.
Here for the selected path and partition, we can visualize a hidden attractor of
system \eqref{sys:lorenz-general} (see Fig.~\ref{fig:gen_lorenz:attr:hidden}).
\begin{figure}[!h]
 \centering
 \includegraphics[width=0.5\textwidth]{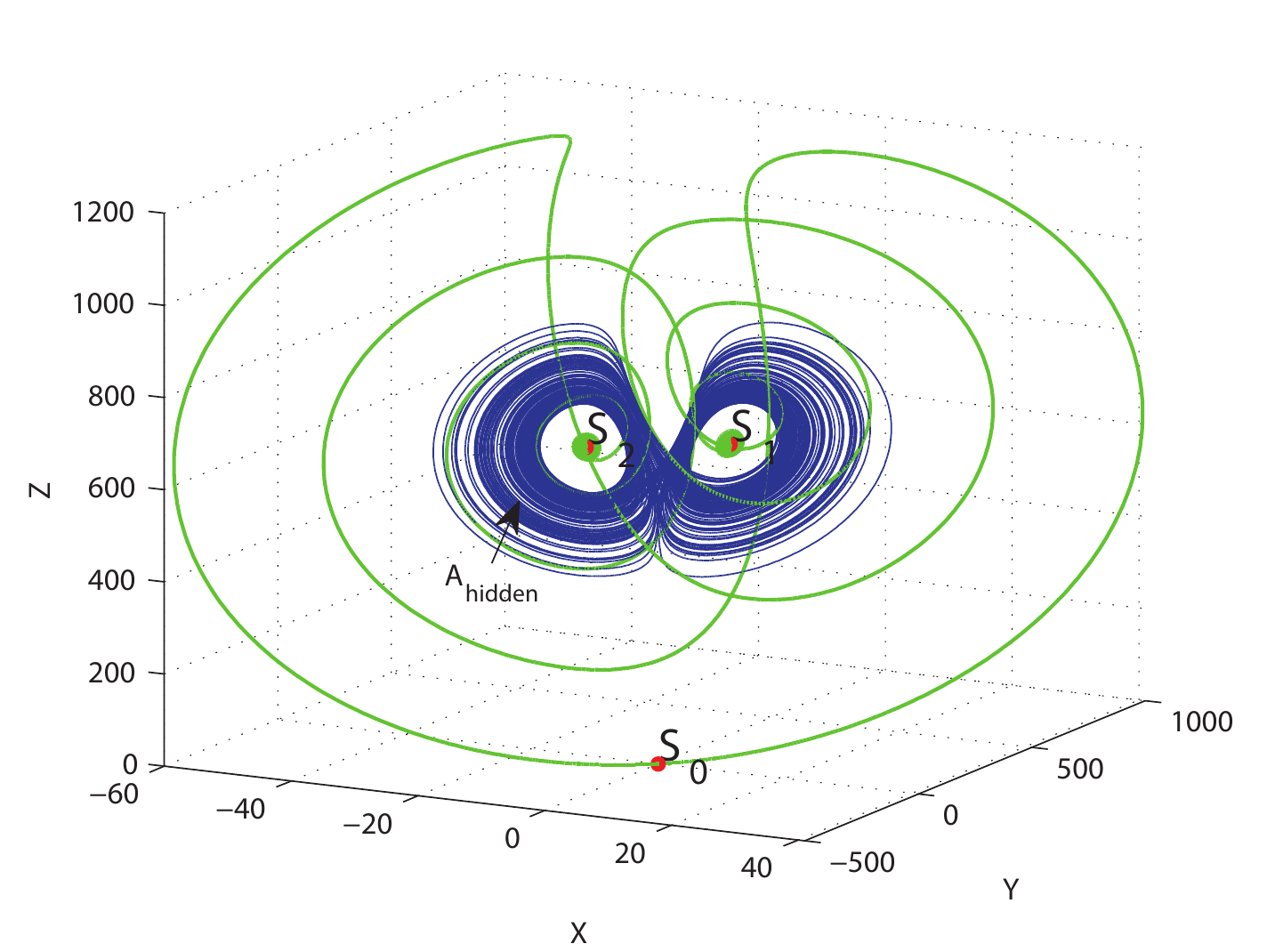}
 \caption{\label{fig:gen_lorenz:attr:hidden}
 Hidden attractor (blue) 
 coexist with B-attractor
 (green outgoing separatrix of the saddle $S_0$ attracted to the red equilibria $S_{1,2}$)}
\end{figure}
The results of continuation procedure are given in \cite{LeonovKM-2015-CNSNS}.

Note that the choice of path and its partitions
in the continuation procedure is not trivial.
For example, a similar procedure does not lead to a hidden attractor
for the following path on the plane $(a,r)$.
Consider $r = 33.51541181$, $a = 0.04735056...$ (the rightmost point on the stability domain),
and take a starting point $P_2$: $r = 33.51541181$, $a = 0.046$ near it (Fig.~\ref{fig:cont_proc:paths}).
If we use the partition step $0.001$, then there are no hidden attractors
after crossing the boundary of the stability domain.
For example, if the end point is $P_3$:
$r = 33.51541181$, $a = 0.048$, there is no chaotic attractor
but only trivial attractors (the equilibria $S_{1,2}$).

\section{Analytical localization of global
attractor via Lyapunov functions}\label{sec:dissip}

In the previous sections, we considered the numerical localization
of various self-excited and hidden attractors of system \eqref{sys:lorenz-general}.
It is natural to question if these attractors (or the union of attractors) are global
(in the sense of Definition~\ref{def:attractor})
or if other coexisting attractors can be found.

The \emph{dissipativity} property is important when
proving the existence of a bounded global attractor for a dynamical system
and gives an analytical localization of the global attractor in the phase space.
The dissipativity of a system, on one hand, proves that there are no
trajectories that tend to infinity as $t \to +\infty$
in the phase space and,
on the other hand, can be used one to determine the boundaries of the domain
that all trajectories enter within a finite time.

\begin{definition}
A set $B_0 \subset \mathbb{R}^n$ is said to be \emph{absorbing}
for dynamical system \eqref{eq:ode}
if for any ${\bf \rm x}_0 \in \mathbb{R}^n$ there exists
$T = T({\bf \rm x}_0)$
such that ${\bf \rm x}(t, {\bf \rm x}_0) \in B_0$ for any $t \geq T$.
\end{definition}
Note that the trajectory ${\bf \rm x}(t, {\bf \rm x}_0)$ with ${\bf \rm x}_0 \in B_0$
may leave $B_0$ for only a finite time before it
returns and stays inside for $t \geq T$.

\begin{remark}
In \cite{Levinson-1944}
the ball $B_R = \left\{{\bf \rm x} \in \mathbb{R}^n : |{\bf \rm x}| < R \right\}$ was regarded as an absorbing set.
In this case, if there exists $R > 0$ such that
\[
 \limsup_{t \to \infty} | {\bf \rm x}(t, {\bf \rm x}_0) | < R,
 \quad \text{for any} ~ {\bf \rm x}_0 \in \mathbb{R}^n,
\]
then it is said that a dynamical system is {\it dissipative in the sense of Levinson }.
$R$ is called {\it a radius of dissipativity}\footnote{Because
 any greater radius also satisfies the definition, the minimal $R$ is of interest
 for the problems of attractor localization and definition of ultimate bound.
}.
\end{remark}

\begin{definition}[\cite{Chueshov-1993,Chueshov-2002-book}]\label{def:abs_set} 
Dynamical system \eqref{eq:ode} is called \emph{(pointwise) dissipative}\footnote{
 Together with the notions of an absorbing set and dissipative system,
 \cite{Ladyzhenskaya-1991,BoichenkoLR-2005} also considered the definitions of a
 {\it B-absorbing set} and a {\it B-dissipative} system
 (uniform convergence of trajectories to the corresponding
 B-absorbing set).
 It is known \cite{BoichenkoLR-2005} that if a dynamical system given on
 $(\mathbb{R}^n, ~ ||\cdot||)$ is dissipative, then it is also B-dissipative.
} if it possesses a bounded absorbing set.
\end{definition}

\begin{theorem}[\cite{Chueshov-1993,Chueshov-2002-book}]
If dynamical system \eqref{eq:ode} is dissipative,
then it possesses a global B-attractor.
\end{theorem}

We can effectively prove dissipativity by
constructing the Lyapunov function \cite{Yoshizawa-1966,LeonovR-1987}.
Consider a sufficient condition of dissipativity for system \eqref{eq:ode}.
\begin{theorem}[\cite{Yakubovich-1964-I,Reitmann-2013}]\label{theorem:dissip1}
 Suppose that there exists continuously differentiable function
 $V({\bf \rm x}) ~ : ~ \mathbb{R}^n \to \mathbb{R}$,
 possessing the following properties.
 \begin{enumerate}[label=(\arabic*)]
 \item
 $ \lim_{|{\bf \rm x}| \to \infty} V({\bf \rm x}) = +\infty $, and
 \item
 there exist numbers $R$ and $\varkappa$ such that for any solution
 ${\bf \rm x}(t, {\bf \rm x}_0)$
 of system \eqref{eq:ode}, the condition $|{\bf \rm x}(t, {\bf \rm x}_0)| > R$
 implies that
 $\dot{V}({\bf \rm x}(t, {\bf \rm x}_0)) \leq -\varkappa$.
 \end{enumerate}

 \medskip
 Then
 \begin{enumerate}[label=(\alph*)]
 \item any solution ${\bf \rm x}(t, {\bf \rm x}_0)$ to \eqref{eq:ode}
 exists at least on
 $[0,+\infty)$, so system \eqref{eq:ode} generates a dynamical
 system for any $t \geq 0$ and ${\bf \rm x}_0 \in \mathbb{R}^n$; and
 \item if $\eta > 0$ is such that
 $B_0 = \{{\bf \rm x} \in \mathbb{R}^n ~|~ V({\bf \rm x}) \leq \eta\}
 \supset \{{\bf \rm x} \in \mathbb{R}^n ~|~ ||{\bf \rm x}||< R\}$,
 then $B_0$ is a compact absorbing set of dynamical system
\eqref{eq:ode}.
 \end{enumerate}
\end{theorem}

More general theorems, connected with the application of the Lyapunov functions to
the proof of dissipativity for dynamical systems can be found in \cite{LeonovBSh-1996,Rasvan-2006}.

It is known that the Lorenz system is dissipative
(it is sufficient to choose the Lyapunov function
$V(x,y,z) = \frac{1}{2}(x^2 + y^2 + (z - r - \sigma)^2)$).
However, for example, one of the Rossler systems is not dissipative in the sense of Levinson \cite{LeonovR-1986}
because the outgoing separatrix is unbounded.
In the general case, there is an art in the construction of Lyapunov functions
which prove dissipativity.


\begin{figure}[ht]
 \centering
 \subfloat[
 {\scriptsize Absorbing set (grey),
 self-excited attractor (blue),
 and global B-attractor
 (blue, green and red) for system \eqref{sys:lorenz-general}
 with the parameters $r = 687.5$, $\sigma = 4$, and $a = 0.0052$.
 $\eta \approx 238541.3$.
 }
 ] {
 \label{fig:gen_lorenz:abs_set:self-exc}
 \includegraphics[width=0.5\textwidth]{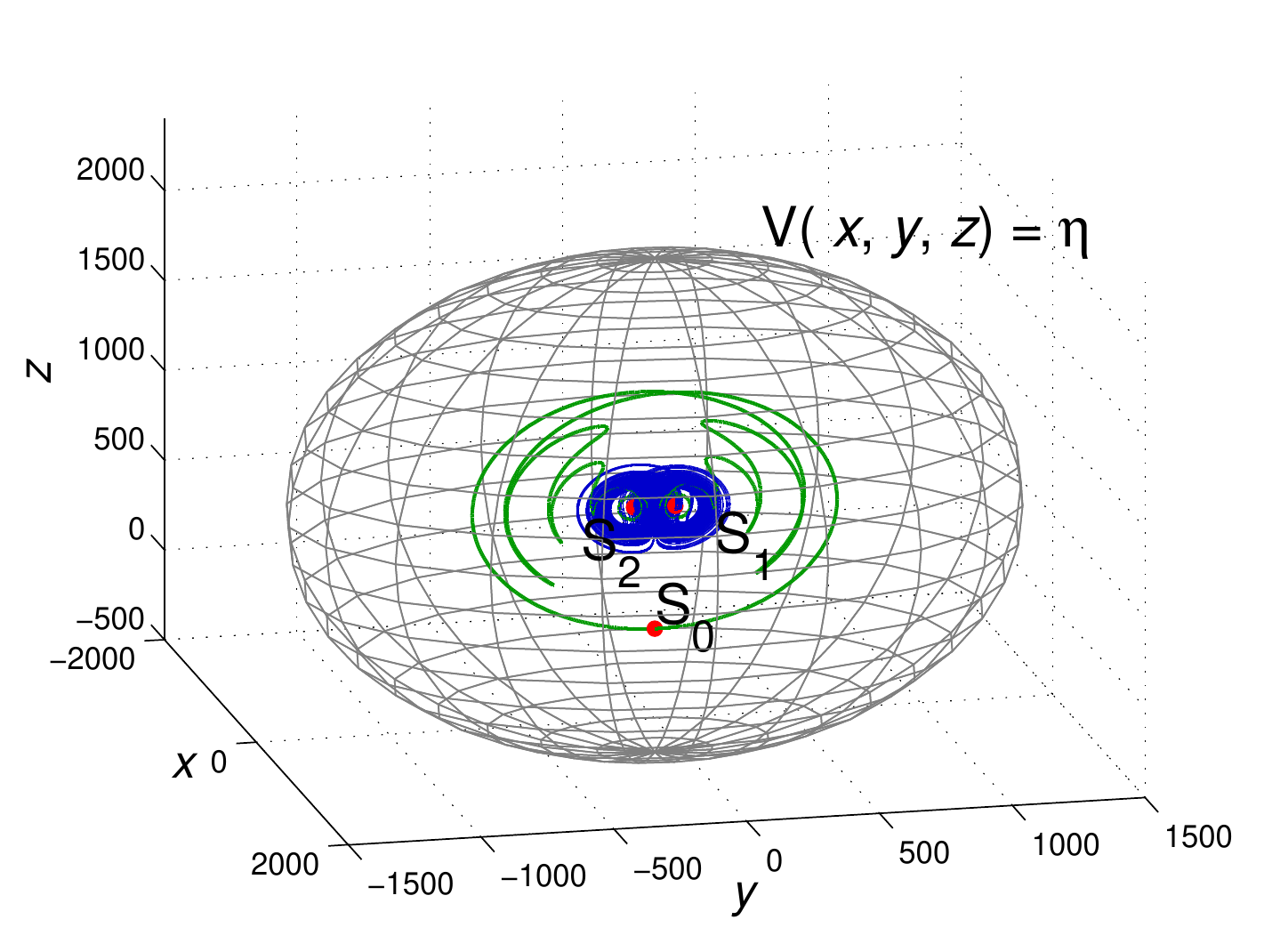}
 }~~
 \subfloat[
 {\scriptsize Absorbing set (grey), hidden attractor (blue),
 and global B-attractor (blue, green and red) for system \eqref{sys:lorenz-general}
 with the parameters $r = 700$, $\sigma = 4$, and $a = 0.0052$.
 $\eta \approx 247230.5$.
 }
 ] {
 \label{fig:gen_lorenz:abs_set:hidden}
 \includegraphics[width=0.5\textwidth]{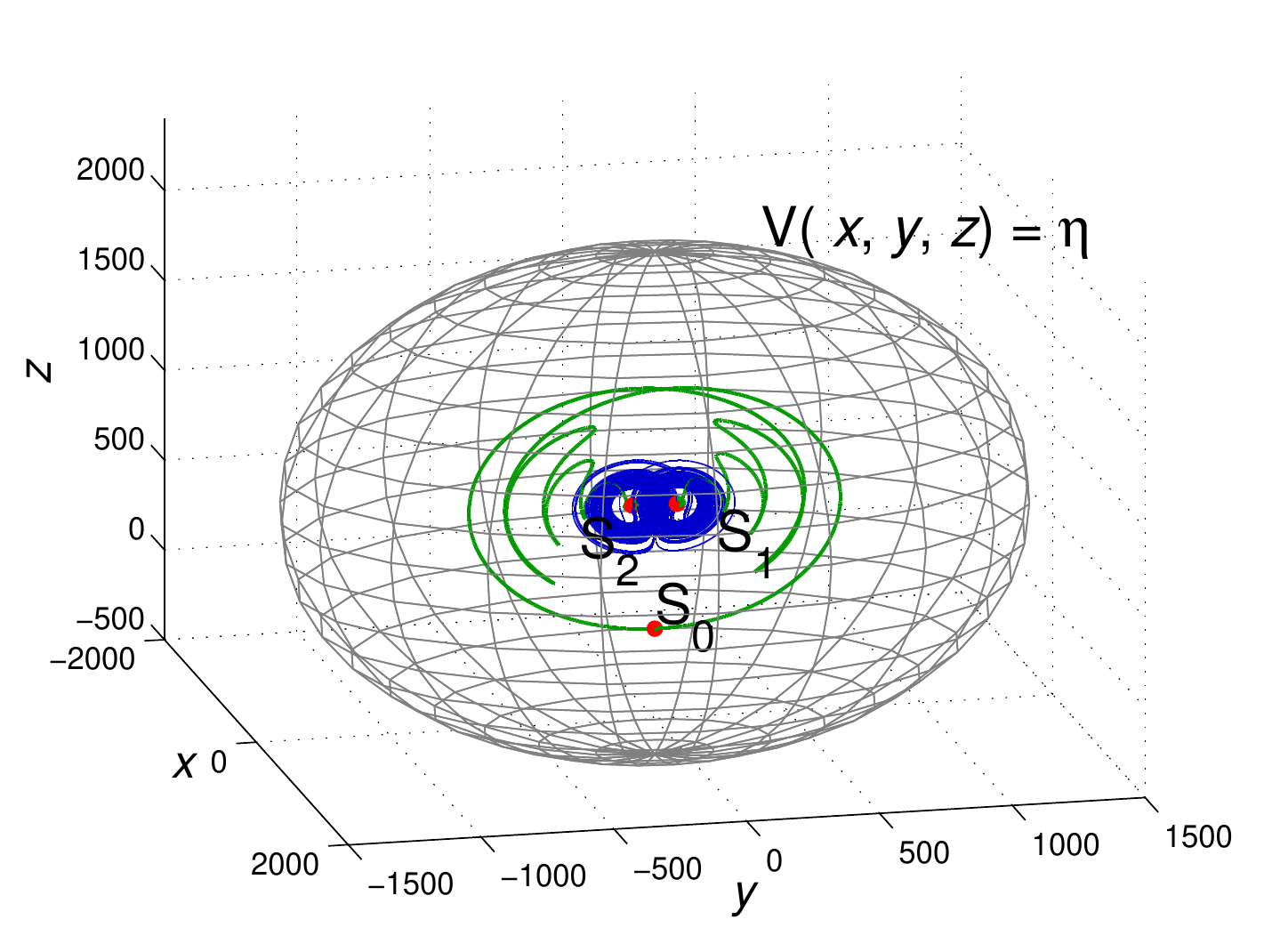}
 }
 \caption{Absorbing sets for system \eqref{sys:lorenz-general}}
 \label{fig:gen_lorenz:abs_set}
\end{figure}

\begin{lemma}\label{theorem:gen_lorenz:dissip}
Dynamical system \eqref{sys:lorenz-general} is dissipative.
\end{lemma}

The proof is based on Lyapunov function $V$ from \eqref{flgd} (see Appendix~\ref{appendix:dissip}).
If $R, \eta$ are chosen
as in the proof of Theorem \ref{theorem:gen_lorenz:dissip}, Appendix~\ref{appendix:dissip},
then dynamical system \eqref{sys:lorenz-general}
has a compact absorbing set
\[
  B_0 = \left\{ (x,y,z):\,
  V(x,y,z) = \frac{1}{2} \left( x^2 + y^2 + (a+1)
 \left(z - \frac{\sigma + r}{a + 1} \right)^2 \right) \leq \eta \right\}.
\]
For example, for $\sigma>1, r>1, a<1$ we can choose
$R= \dfrac{\sigma+r}{a+1}$ and $\eta= 2(a+1)R^2$
(see Fig. \ref{fig:gen_lorenz:abs_set}).
Note that for system \eqref{sys:lorenz-general}
the ellipsoidal absorbing set $B_0$
can be improved using special additional transformations
and Yudovich's theorem (see, e.g., \cite{Belykh-1980}),
similarly to \cite{Leonov-1983} for the Lorenz system.

There is also a cylindrical positively invariant set
for system \eqref{sys:lorenz-general} \cite{Leonov-2014-ND},
\begin{equation}\label{pinv}
  C = \{|x| \leq r + \frac{2a}{\sigma}r^2,\quad  y^2 + (z -r)^2 \leq r^2\},
\end{equation}
because
\[
     (y^2 + (z -r)^2)^{\bullet} \leq -(y^2 + (z -r)^2) + r^2 < 0 \quad \forall x,\ |y| > r \text{\ or\ } |z| > 2r
\]
and
\[
  |x|^{\bullet} \leq -\sigma|x|+\sigma|y|+a|y||z| < 0 \quad |y|\leq r,\, |z| \leq 2r,\,|x| > r+2ar^2/\sigma.
\]
Thus, as for the Lorenz system \cite{LeonovBK-1987},
we obtain three different estimates of the attractor:
the ball $B_{R}$, ellipse $B_0$, and cylinder $C$.

\section{Upper estimate of the Lyapunov dimension of attractor}

\subsection{Lyapunov exponents and Lyapunov dimension}

Suppose that the right-hand side of system \eqref{eq:ode} is sufficiently smooth,
and consider a linearized system along a solution ${\bf \rm x}(t, {\bf \rm x}_0)$.
We have
\begin{equation}\label{eq:lin_eq}
 \dot{{\bf \rm u}} = J({\bf \rm x}(t, {\bf \rm x}_0)) \, {\bf \rm u},
 \quad {\bf \rm u} \in \mathbb{R}^n, ~t \in \mathbb{R}_+,
\end{equation}
where
\[
J({\bf \rm x}(t, {\bf \rm x}_0)) =
\left[
 \frac{\partial f_i({\bf \rm x})}{\partial {\bf \rm x}_j}\big|_{{\bf \rm x}
= {\bf \rm x}(t, {\bf \rm x}_0)}
\right]
\]
is the $(n \times n)$ Jacobian matrix evaluated along the trajectory
 ${{\bf \rm x}(t, {\bf \rm x}_0)}$ of system \eqref{eq:ode}.
A fundamental matrix $X(t,{\bf \rm x}_0) $ of linearized system \eqref{eq:lin_eq}
is defined by the variational equation
\begin{equation}\label{eq:var_eq}
 \dot{X}(t,{\bf \rm x}_0) =
 J({\bf \rm x}(t, {\bf \rm x}_0)) \, X(t,{\bf \rm x}_0).
\end{equation}
We typically set $X(0,{\bf \rm x}_0) = I_n$, where $I_n$ is the
identity matrix.
Then
${\bf \rm u}(t, {\bf \rm u}_0) = X(t,{\bf \rm x}_0) {\bf \rm u}_0$.
In the general case,
${\bf \rm u}(t, {\bf \rm u}_0) = X(t,{\bf \rm x}_0)
X^{-1}(0,{\bf \rm x}_0){\bf \rm u}_0$.
Note that if a solution of nonlinear system \eqref{eq:ode}
is known, then we have
\[
 X(t,{\bf \rm x}_0) =
 \frac{\partial {\bf \rm x}(t, {\bf \rm x}_0)}{\partial {\bf \rm x}_0}.
\]

 Two well-known definitions of Lyapunov exponents are
 the upper bounds of the exponential growth rate
 of the norms of linearized system solutions (LCEs) \cite{Lyapunov-1892}
 and the upper bounds of the
 exponential growth rate of the singular values
 of fundamental matrix of linearized system (LEs) \cite{Oseledec-1968}.

Let $\sigma_1 (X(t,{\bf \rm x}_0)) \geq \cdots
\geq \sigma_n (X(t,{\bf \rm x}_0)) > 0$ denote
the singular values of a fundamental matrix $X(t,{\bf \rm x}_0)$
(the square roots of the eigenvalues of the matrix
$X(t,{\bf \rm x}_0)^{*}X(t,{\bf \rm x}_0)$ are reordered for each $t$).

\begin{definition}\label{def:le}
The {\it Lyapunov exponents (LEs)} at the point ${\bf \rm x}_0$
are the numbers (or the symbols $\pm \infty$) defined by
\begin{equation}\label{defLE}
  {\rm LE}_i ({\bf \rm x}_0) = \limsup_{t \to \infty} \frac{1}{t}
  \ln \sigma_i (X(t,{\bf \rm x}_0)).
\end{equation}
\end{definition}

 LEs are commonly used\footnote{
 The  LCEs \cite{Lyapunov-1892} and LEs \cite{Oseledec-1968} are ``often'' equal,
 e.g., for a ``typical'' system that satisfies the conditions
 of Oseledec theorem \cite{Oseledec-1968}.
 However, there are no effective methods for checking
 Oseledec conditions for a given system:
 ``{\it Oseledec proof is important mathematics,
 but the method is not helpful in elucidating dynamics}'' \cite[p.118]{ChaosBook}).
 For a particular system, LCEs and LEs may be different.
 For example,
 for the fundamental matrix
 \(
    X(t)=\left(
      \begin{array}{cc}
        1 & g(t)-g^{-1}(t) \\
        0 & 1 \\
      \end{array}
    \right)
 \) we have the following ordered values:
 $  \LCE_1 =
  {\rm max}\big(\limsup\limits_{t \to +\infty}\mathcal{X}[g(t)],
  \limsup\limits_{t \to +\infty}\mathcal{X}[g^{-1}(t)]\big),
  \LCE_2 = 0$;
 $
  \LE_{1,2} = {\rm max, min}
  \big(
     \limsup\limits_{t \to +\infty}\mathcal{X}[g(t)],
     \limsup\limits_{t \to +\infty}\mathcal{X}[g^{-1}(t)]
  \big)
 $, where $\mathcal{X}(\cdot) = \frac{1}{t}\log|\cdot|$.
 Thus, in general, the Kaplan-Yorke (Lyapunov) dimensions based
 on LEs and LCEs may be different.
 Note also that positive largest LCE or LE, computed via the linearization of the system along a trajectory,
 do not necessary imply instability or chaos,
 because for non-regular linearization
 there are well-known Perron effects of Lyapunov exponent sign reversal
 \cite{LeonovK-2007,KuznetsovL-2001,KuznetsovL-2005}.
 Therefore for the computation of the Lyapunov dimension of an attractor
 one has to consider a grid of points on the attractor and
 corresponding local Lyapunov dimensions \cite{KuznetsovMV-2014-CNSNS}.
 More detailed discussion and examples can be found
 in \cite{KuznetsovAL-2014-arXiv-LE,LeonovK-2007}.
}
in the theories of dynamical systems and attractor dimensions
\cite{Ledrappier-1981,EckmannR-1985,Hunt-1996,Temam-1997,BoichenkoLR-2005,BarreiraG-2011}.

\begin{remark}
 The LEs are independent of the choice  of fundamental matrix at the point ${\bf \rm x}_0$ \cite{KuznetsovAL-2014-arXiv-LE}
 unlike the Lyapunov characteristic exponents (LCEs, see \cite{Lyapunov-1892}).
 To determine all possible values of LCEs, we must consider a \emph{normal fundamental matrix}.
\end{remark}

We now define a Lyapunov dimension \cite{KaplanY-1979}

\begin{definition}
A local Lyapunov dimension of a point ${\bf \rm x}_0$
in the phase space of a dynamical system is as follows:
$\dim_L {\bf \rm x}_0 = 0$ if $\LE_1(x_0) \leq 0$ and $\dim_L {\bf \rm x}_0=n$ if $\sum_{i=1}^{n}\LE^{o}_i(x_0) \geq 0$,
otherwise
\begin{equation}\label{formula:kaplan}
 \dim_L {\bf \rm x}_0 = j({\bf \rm x}_0) +
 \cfrac{{\rm LE}_1({\bf \rm x}_0) + \ldots +
 {\rm LE}_j({\bf \rm x}_0)}{|{\rm LE}_{j+1}({\bf \rm x}_0)|},
\end{equation}
where ${\rm LE}_1({\bf \rm x}_0) \geq \ldots \geq {\rm LE}_n({\bf \rm x}_0)$
are ordered LEs and
$j({\bf \rm x}_0) \in [1, n]$ is the smallest natural number $m$ such that
$$
{\rm LE}_1({\bf \rm x}_0) + \ldots + {\rm LE}_{m}({\bf \rm x}_0) > 0,
\quad {\rm LE}_{m+1}({\bf \rm x}_0) < 0, \quad
\cfrac{{\rm LE}_1({\bf \rm x}_0) + \ldots +
{\rm LE}_m({\bf \rm x}_0)}{|{\rm LE}_{m+1}({\bf \rm x}_0)|} < 1.
$$
\end{definition}

The Lyapunov dimension of invariant set $K$ of a dynamical system is defined as
\begin{equation}
 \dim_L K = \sup_{x_0 \in K} \dim_Lx_0.
\end{equation}

Note that, from an applications perspective,
an important property of the Lyapunov dimension is
the chain of inequalities~\cite{Hunt-1996,IlyashenkoW-1999,BoichenkoLR-2005}
\begin{equation}
 \dim_T K \leqslant \dim_H K \leqslant \dim_F K
 \leqslant \dim_L K \label{ineq:dimensions}.
\end{equation}
Here $\dim_T K, \dim_H K,$ and $\dim_F K$ are the topological,
Hausdorff, and fractal dimensions of $K$, respectively.

Along with commonly used numerical methods for estimating and computing the Lyapunov dimension,
there is an analytical approach that was proposed by Leonov
\cite{Leonov-1991-Vest,LeonovB-1992,BoichenkoLR-2005,Leonov-2008,Leonov-2012-PMM,LeonovK-2015-AMC}.
It is based on the direct Lyapunov method
and uses Lyapunov-like functions.

LEs and the Lyapunov dimension are invariant under linear
changes of variables (see, e.g., \cite{KuznetsovAL-2014-arXiv-LE}).
Therefore we can apply the linear variable change
${\bf\rm y} = S {\bf\rm x}$ with a nonsingular $n \times n$-matrix $S$.
Then system \eqref{eq:ode} is transformed into
\[
 \dot{{\bf\rm y}} = S \,\dot{{\bf\rm x}} = S \,{\bf\rm f}
 (S^{-1}{\bf\rm y}) = \tilde{{\bf\rm f}}({\bf\rm y}).
\]
Consider the linearization along corresponding solution
${\bf\rm y}(t, {\bf\rm y}_0) = S {\bf\rm x}(t, S^{-1} {\bf\rm x}_0)$,
that is,
\begin{equation}\label{eq:lin_eq-new}
 \dot {\bf \rm v} = \tilde{J}({\bf \rm y}(t, {\bf \rm y}_0))\,{\bf \rm v},
 \quad {\bf \rm v} \in \mathbb{R}^n.
\end{equation}
Here the Jacobian matrix is as follows
\begin{align}
 \tilde{J}({\bf\rm y}(t, {\bf\rm y}_0))
 =S \, J({\bf\rm x}(t, {\bf\rm x}_0)) \, S^{-1} \label{jacobian-new}
\end{align}
and the corresponding fundamental matrix satisfies
\(
   Y(t,{\bf \rm y}_0) = S X(t,{\bf \rm x}_0).
\)

For simplicity, let $J({\bf \rm x}) = J({\bf \rm x}(t, {\bf \rm x}_0))$.
Suppose that
$\lambda_1 ({\bf \rm x},S) \geqslant \cdots \geqslant \lambda_n ({\bf \rm x},S)$
are eigenvalues of the symmetrized Jacobian matrix \eqref{jacobian-new}
\begin{equation}
 \frac{1}{2} \left( S J({\bf \rm x}) S^{-1} +
 (S J({\bf \rm x}) S^{-1})^{*}\right).
 \label{SJS}
\end{equation}

\begin{theorem}[\cite{Leonov-2002,Leonov-2012-PMM}]\label{theorem:th1}
Given an integer $j \in [1,n]$ and $s \in [0,1]$,
suppose that there are a continuously
differentiable scalar function $\vartheta: \mathbb{R}^n \rightarrow \mathbb{R}$
and a nonsingular matrix $S$ such that
\begin{equation}\label{ineq:th-1}
 \lambda_1 ({\bf \rm x},S) + \cdots + \lambda_j ({\bf \rm x},S) + s\lambda_{j+1}
 ({\bf \rm x},S) + \dot{\vartheta}({\bf \rm x}) < 0,
 ~ \forall \, {\bf \rm x} \in K.
\end{equation}
Then $\dim_L K \leqslant j+s$.
\end{theorem}
Here $\dot{\vartheta}$ is the derivative of $\vartheta$ with respect
to the vector field ${\bf\rm f}$:
$$
 \dot{\vartheta} ({\bf \rm x}) = ({\rm grad}(\vartheta))^{*}{\bf\rm f}({\bf \rm
x}).
$$
The introduction of the matrix $S$ can be regarded as a change of the space metric.

\begin{theorem}[\cite{Leonov-1991-Vest,LeonovB-1992,BoichenkoLR-2005,Leonov-2012-PMM}]
\label{theorem:th2}
Assume that there are a continuously differentiable scalar function $\vartheta$
and a nonsingular matrix $S$ such that
\begin{equation}\label{ineq:th-2}
 \lambda_1 ({\bf \rm x},S) + \lambda_2 ({\bf \rm x},S) +
 \dot{\vartheta}({\bf \rm x}) < 0,
 ~ \forall \, {\bf \rm x} \in \mathbb{R}^n.
\end{equation}
Then any solution of system \eqref{eq:ode} bounded on $[0,+\infty)$
tends to an equilibrium as $t \rightarrow +\infty$.
\end{theorem}
Thus, if \eqref{ineq:th-2} holds,
then the global attractor of system \eqref{eq:ode}
coincides with its stationary set.

Theorems \ref{theorem:th1} and \ref{theorem:th2}
give the following results for system~\eqref{sys:lorenz-general}.

\begin{theorem}\label{theorem:th3}
Suppose that $\sigma > 1$. \\
If
\begin{equation}
 \left(r+\frac{\sigma}{a}\right)^2 <
 \frac{2(\sigma + 1)}{a}, \label{cond:param1}
\end{equation}
then any solution of system~\eqref{sys:lorenz-general}
bounded on $\left[0,+\infty\right)$
tends to an equilibrium as $t \rightarrow +\infty$. \\
If
\begin{equation}
 \left(r+\frac{\sigma}{a}\right)^2 >
 \frac{2(\sigma + 1)}{a}, \label{cond:param2}
\end{equation}
then
\begin{equation}
 \dim_L{K} \leqslant 3 - \frac{2(\sigma + 2)}{\sigma + 1
 + \sqrt{(\sigma - 1)^2 + a\left(\frac{\sigma}{a}+r\right)^2}}.
 \label{ineq:lyap-dim}
\end{equation}

\end{theorem}
\begin{proof}

We use the matrix
\begin{equation*}
 S = \left(
 \begin{array}{ccc}
 -a^{-\frac{1}{2}} & 0 & 0 \\
 0 & 1 & 0 \\
 0 & 0 & 1
 \end{array}
 \right).
\end{equation*}
Then the eigenvalues of the corresponding matrix \eqref{SJS}
are the following
\begin{align*}
	& \lambda_2 = -1, & \\
	& \lambda_{1,3} = -\frac{\sigma+1}{2} \pm
	\frac{1}{2}\left[(\sigma - 1)^2 +
	a \left(2 z - \frac{\sigma + a r}{a}\right)^2
	\right]^{\frac{1}{2}}. &
\end{align*}
To check property \eqref{ineq:th-1} of Theorem \ref{theorem:th1}
and property \eqref{ineq:th-2} of Theorem \ref{theorem:th2},
we can consider the Lyapunov-like function
\begin{equation*}
 \vartheta(x,y,z) = \frac{2(1-s) V(x,y,z)}{\left[(\sigma - 1)^2
 + a\left(\frac{\sigma}{a}+r\right)^2 \right]^{\frac{1}{2}}},
\end{equation*}
where
\begin{equation*}
 V(x,y,z) = \frac{\gamma}{\sigma} x^2 + \gamma y^2 +
 \gamma \left( 1 + \frac{a}{\sigma}\right)z^2 - 2 \gamma (r-1) z,
 \quad \gamma = \frac{\sigma + a r}{2(r-1)}.
\end{equation*}

Finally, for system \eqref{sys:lorenz-general}
with given $S$ and $\vartheta$,
if condition \eqref{cond:param2} is satisfied and
\begin{equation*}
 s > \frac{-(\sigma + 3) + \sqrt{(\sigma - 1)^2
 + a\left(\frac{\sigma}{a}+r\right)^2}}
 {\sigma + 1 + \sqrt{(\sigma - 1)^2 +
 a\left(\frac{\sigma}{a}+r\right)^2}}, \label{ineq:s}
\end{equation*}
then Theorem \ref{theorem:th1} gives \eqref{ineq:lyap-dim}.
If condition \eqref{cond:param1} is valid and $s = 0$,
then the conditions of Theorem \ref{theorem:th2} are satisfied
and any solution bounded on $[0, +\infty)$ tends to
an equilibrium as $t \to +\infty$. \qed
\end{proof}


Note that for  $\sigma = 4$, $r = 687.5$, and $a = 0.0052$
the analytical estimate of the Lyapunov dimension of the corresponding self-excited attractor
is as follows
\[
 \dim_L K < 2.890997461...
\]
and the values of the local Lyapunov dimension at equilibria are
\[
 \dim_L {S_0} = 2.890833450..., \quad \dim_L {S_{1,2}} = 2.009763700... .
\]
Numerically, by an algorithm in Appendix~\ref{appSVD},
the Lyapunov dimension of the self-excited attractor is  ${\rm LD} = 2.1405$.

The analytical estimate of the Lyapunov dimension of the hidden attractor for  $\sigma = 4$, $r = 700$, and $a = 0.0052$ is as follows
\[
 \dim_L K < 2.891882349...,
\]
and the local Lyapunov dimension at the stationary points
are the following
\[
 \dim_L {S_0} = 2.891767634..., \dim_L {S_{1,2}} = 1.966483617...
\]
Numerically, the Lyapunov dimension of the hidden attractor is ${\rm LD} = 2.1322$.

Thus, the Lyapunov dimensions of B-attractor (which involve equilibrium $S_0$)
and the global attractor are very close to the analytical estimate.

In the general case the coincidence of the analytical upper estimate
with the local Lyapunov dimension at a stationary point
gives the exact value of the Lyapunov dimension of the global attractor
(see, e.g., studies of various Lorenz-like systems
\cite{Leonov-1991-Vest,LeonovB-1992,LeonovPS-2011,Leonov-2012-PMM,LeonovKKK-2015,LeonovK-2015-AMC}).

\newpage
\appendix
\section*{Appendix}
\section{Description of the physical problem}
\label{appendix:phys-problem}

Consider the convection of viscous incompressible fluid motion inside the ellipsoid
\[
 \left(\frac{x_1}{a_1}\right)^2 + \left(\frac{x_2}{a_2}\right)^2 +
 \left(\frac{x_3}{a_3}\right)^2 = 1,
 \quad a_1 > a_2 > a_3 > 0
\]
under the condition of stationary inhomogeneous external heating.
We assume that the ellipsoid and heat sources rotate
with constant velocity ${\bf \Omega_0}$ around the axis.
Vector ${\bf l_0}$ determines the orientation of the ellipsoid
and has the same direction as the gravity vector ${\bf g}$.
Vector ${\bf g}$ is stationary with respect to the ellipsoid motion.
The value ${\bf \Omega_0}$ is assumed to be such that
the centrifugal forces can be neglected when
compared with the influence of the gravitational field.
Consider the case when the ellipsoid rotates around the axis $x_3$
that has a constant angle $\alpha$ with gravity vector
${\bf g}$ ($|{\bf g}| = g$).
The vector ${\bf g}$ is placed in the plane $x_1 x_3$.
Then, ${\bf \Omega_0} = (0,\,0,\,\Omega_0)$ and
${\bf l_0} = (a_1 \sin \alpha, \, 0, \, -a_3 \cos \alpha)$.
Let the steady-state temperature difference
$\Delta {\bf \hat{T}} = (q_0, \, 0, \, 0)$
be generated along the axis $x_1$ (Fig.~\ref{fig:ellipsoid}).
The corresponding mathematical model (three-mode model of convection)
was obtained by Glukhovsky and Dolzhansky \cite{GlukhovskyD-1980}
in the form (see \eqref{sys:conv_fluid})
\begin{equation*}
\begin{cases}
\dot{x} $ = $ - \sigma x + C z + A y z, \\
\dot{y} $ = $ R_a - y -x z, \\
\dot{z} $ = $ - z + x y.
\end{cases}
\end{equation*}

Here
\begin{align*}
 \sigma &= \frac{\lambda}{\mu}, \qquad T_a = \frac{\Omega_0^2}{\lambda^2}, \quad
 R_a = \frac{g \beta a_3 q_0}{2 a_1 a_2 \lambda \mu}, & \\
 A &= \frac{a_1^2 - a_2^2}{a_1^2 + a_2^2} \cos^2 \alpha \, T_a^{-1}, \quad
 C = \frac{2 a_1^2 a_2}{a_3 (a_1^2 + a_2^2)} \sigma \sin \alpha, &\\
 x(t) &= \mu^{-1} \left(\omega_3(t) + \frac{g \beta a_3 \cos \alpha}{2 a_1 a_2 \Omega_0} q_3(t) \right), \quad
 y(t) = \frac{g \beta a_3}{2 a_1 a_2 \lambda \mu} q_1(t), &\\
 z(t) &= \frac{g \beta a_3}{2 a_1 a_2 \lambda \mu} q_2(t), &
\end{align*}
and $\lambda, \mu, \beta$ are the coefficients of viscosity,
heat conduction, and volume expansion, respectively;
$q_{1}(t)$, $q_{2}(t)$, and $q_{3}(t)$
($q_3(t) \equiv 0$) are
temperature differences on the principal axes of the ellipsoid;
$\omega_1 (t)$, $\omega_2 (t)$, and $\omega_3 (t)$ are the projections
of the vectors of fluid angular velocities
on the axes $x_1$, $x_2$, and $x_3$, respectively. Here
\[
 \omega_1(t) = - \frac{g \beta a_3}{2 a_1 a_2 \Omega_0} \cos \alpha \, q_1(t), \quad
 \omega_2(t) = - \frac{g \beta a_3}{2 a_1 a_2 \Omega_0} \cos \alpha \, q_2(t).
\]
The parameters $\sigma$, $T_a$, and $R_a$ are the Prandtl, Taylor, and Rayleigh numbers, respectively.

\begin{figure}[h!]
 \centering
 \includegraphics[width=0.25\textwidth]{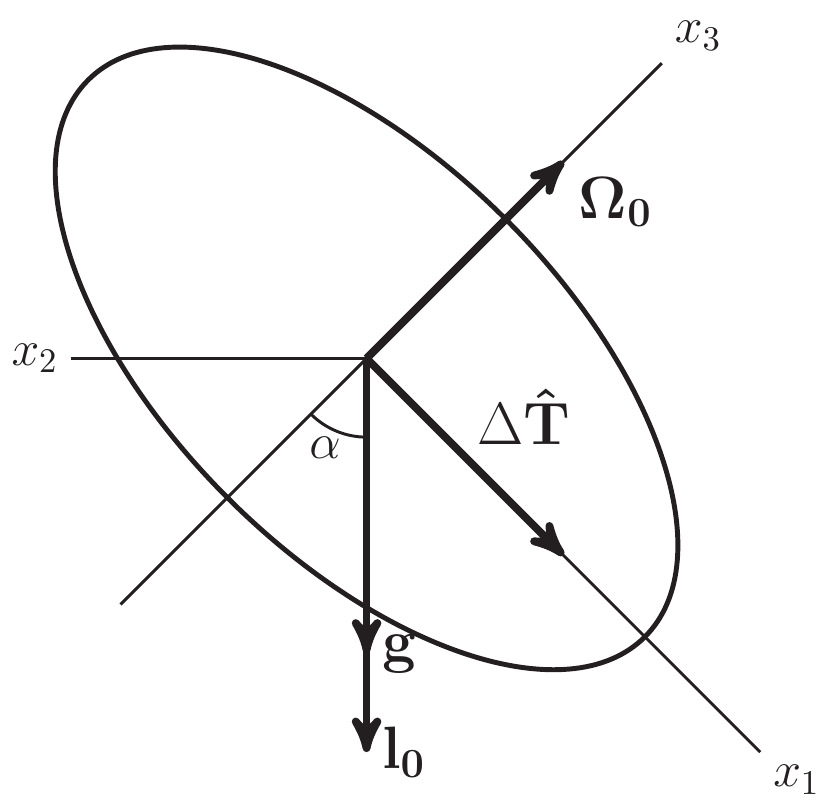}
 \caption{Illustration of the problem}
 \label{fig:ellipsoid}
\end{figure}

The linear change of variables \cite{GlukhovskyD-1980}
\begin{equation*}
   x \to x, \quad y \to C^{-1} y, \quad z \to C^{-1} z,
\end{equation*}
transforms system \eqref{sys:conv_fluid} into the system
\begin{equation}\label{sys:glukh-golzh}
\begin{cases}
\dot{x} $ = $ - \sigma x + z + A_c y z, \\
\dot{y} $ = $ R_c - y -x z, \\
\dot{z} $ = $ - z + x y.
\end{cases}
\end{equation}
with
\[
	R_c = R_a C, \quad A_c = \frac{A}{C^2}.
\]

After the linear transformation (see, e.g., \cite{LeonovB-1992}):
\begin{equation}\label{sys:glukh-dolzh:change_var}
   x \to x, \quad y \to R_c - \frac{\sigma}{R_c A_c + 1} z,
   \quad z \to \frac{\sigma}{R_c A_c + 1} y,
\end{equation}
system \eqref{sys:glukh-golzh} takes the form of \eqref{sys:lorenz-general}
with
\begin{equation}
	a = \frac{A_c \sigma^2}{(R_c A_c + 1)^2},
	\quad r = \frac{R_c}{\sigma}(R_c A_c + 1).
	\label{sys:glukh-dolzh:change_var:param}	
\end{equation}

\section{Proof of Proposition \ref{proposition:gen_lorenz:stability}}
\label{appendix:stability}


For system \eqref{sys:lorenz-general}
the characteristic polynomial of the Jacobian matrix of the right-hand side
at the point ${\bf \rm x}_0 = (x_0,y_0,z_0) \in \mathbb{R}^3$
has the form
\[
\chi({\bf \rm x}_0) = \lambda^3 + p_1({\bf \rm x}_0)
\lambda^2 + p_2({\bf \rm x}_0) \lambda + p_3({\bf \rm x}_0),
\]
where
\begin{align*}
& p_1 ({\bf \rm x}_0) = \sigma + 2, \\
& p_2 ({\bf \rm x}_0) = x_0^2 + a y_0^2 - a z_0^2 +
 (\sigma + a r)z_0 - r \sigma + 2 \sigma + 1, \\
& p_3 ({\bf \rm x}_0) = \sigma x_0^2 + a y_0^2 -
a z_0^2 - 2 a x_0 y_0 z_0 + (\sigma + a r) x_0 y_0 +
(\sigma + a r) z_0 - r \sigma + \sigma.
\end{align*}

Applying the Hurwitz criterion, 
the necessary and sufficient stability conditions for stationary point ${\bf \rm x}_0$
are the following
\begin{align}
 p_1({\bf \rm x}_0) &> 0, \\
 p_2({\bf \rm x}_0) &> 0, \label{eq:hurw_crit:cond1} \\
 p_3({\bf \rm x}_0) &> 0, ~ \text{and} \label{eq:hurw_crit:cond2} \\
 p_1({\bf \rm x}_0) p_2({\bf \rm x}_0) - p_3({\bf \rm x}_0) &> 0.
 \label{eq:hurw_crit:cond3}
\end{align}

Equilibria $S_{1,2}$ exist for $r>1$, and we can check that $\chi (S_1) = \chi (S_2)$.
For further analysis we can introduce
\begin{equation}
 D = a \left(\sqrt{(\sigma - a r)^2 + 4 \sigma a} - (\sigma - a r)\right) > 0.
\end{equation}
Then for stationary points \eqref{eq:equil_s12}, condition
\eqref{eq:hurw_crit:cond1} takes the form
\begin{equation}
 p_2 (S_{1,2}) = \frac{2}{D}\left( C_1 \sqrt{(\sigma - a r)^2 + 4 \sigma a} -
 C_1 (\sigma - a r) - 2 \sigma^2 a (\sigma - a r) \right) > 0,
 \label{eq:hurw_crit:cond1:s12}
\end{equation}
where
$$
 C_1 = \sigma (\sigma - a r)^2 + a^2 r + a \sigma^2 > 0.
$$
Because $\sigma > a r > 0$, we have
\begin{align*}
 C_1 \sqrt{(\sigma - a r)^2 + 4 \sigma a} &> C_1 (\sigma - a r)
 + 2 \sigma^2 a (\sigma - a r) \quad \text{iff}\\
 (\sigma - a r)^2 + 4 \sigma a &> \left((\sigma - a r) +
 \frac{2 \sigma^2 a (\sigma - a r)}{C_1} \right)^2
 \quad \text{iff} \\
 4 \sigma a &> \frac{4 \sigma^2 a (\sigma - a r)^2}{C_1} +
 \frac{4 \sigma^4 a^2 (\sigma - a r)^2}{C_1^2}
 \quad \text{iff}\\
 C_1^2 &> \sigma (\sigma - a r)^2 C_1 + \sigma^3 a (\sigma - a r)^2.
\end{align*}
The last inequality is satisfied because
\[
 \frac{1}{a} \left( C_1^2 - \sigma (\sigma - a r)^2 C_1 -
 \sigma^3 a (\sigma - a r)^2 \right) =
 \sigma a r (\sigma - a r)^2 + a (\sigma^2 + a r)^2 > 0.
\]
This implies \eqref{eq:hurw_crit:cond1:s12}.

Condition \eqref{eq:hurw_crit:cond2} for $S_{1,2}$ takes the form
\begin{align}
 p_3 (S_{1,2}) = &\frac{2\sigma}{D} \left( \sqrt{(\sigma - a r)^2
 + 4 \sigma a} - (\sigma - a r + 2 a)\right) \cdot
 \nonumber\\
 &\cdot \left((\sigma - a r)^2 + 4 \sigma a - (\sigma - a r)
 \sqrt{(\sigma - a r)^2 + 4 \sigma a}) \right) > 0.
 \label{eq:hurw_crit:cond2:s12}
\end{align}
Since
\[
 (\sigma - a r)^2 + 4 \sigma a - (\sigma - a r + 2 a)^2 = 4 a^2 (r - 1) > 0
\]
and
\[
 (\sigma - a r)^2 + 4 \sigma a > (\sigma - a r) \sqrt{(\sigma - a r)^2
 + 4 \sigma a} \quad \text{iff} \quad
 \sqrt{(\sigma - a r)^2 + 4 \sigma a} > (\sigma - a r),
\]
condition \eqref{eq:hurw_crit:cond2:s12} is also satisfied.

Condition \eqref{eq:hurw_crit:cond3} for stationary points $S_{1,2}$ is as follows
\begin{align}
p_1 (S_{1,2}) p_2 (S_{1,2}) - p_3 (S_{1,2}) = &
\frac{2}{D} \left(C_2 \sqrt{(\sigma - a r)^2 + 4 \sigma a} - \right.\nonumber\\
&\left.- C_2 (\sigma - a r) - 2\sigma^2 a (\sigma (\sigma - a r) - 4 a)\right)>0,
 \label{eq:hurw_crit:cond3:s12}
\end{align}
where
\[
 C_2 = \left((\sigma (\sigma - a r) - a)^2 + a^2 (r \sigma + 2 r - 1)
+ a \sigma^2 (\sigma - 2) \right).
\]
If $\sigma > 2$, then $C_2 > 0$
and we can derive a chain of inequalities for \eqref{eq:hurw_crit:cond3:s12}:
\begin{align*}
 C_2 \sqrt{(\sigma - a r)^2 + 4 \sigma a} &> C_2 (\sigma - a r)
 + 2\sigma^2 a (\sigma (\sigma - a r) - 4a)
 \quad \text{iff} \\
 (\sigma - a r)^2 + 4 \sigma a &> \left((\sigma - ar) +
 \frac{2\sigma^2 a (\sigma (\sigma - a r) - 4 a)}{C_2}\right)^2 \quad
 \text{iff} \\
 4 \sigma a &> \frac{4\sigma^2 a (\sigma - a r) (\sigma (\sigma - a r) - 4a)}{C_2}
 + \frac{4\sigma^4 a^2 (\sigma (\sigma - a r) - 4 a)^2}{C_2^2} \quad \text{iff} \\
 C_2^2 &> \sigma (\sigma - a r) (\sigma (\sigma - a r) - 4 a) C_2 + \sigma^3 a
 (\sigma (\sigma - a r) - 4 a)^2.
\end{align*}
We can divide the last inequality by $\left(-a^2\right)$
and rewrite it in the form of polynomial
\begin{align*}
 a^2 \sigma^2 (\sigma -2) r^3 - a \left(2\sigma^4 - 4\sigma^3 - 3 a
 \sigma^2 + 4 a \sigma + 4 a\right) r^2
 + \sigma^2 \left(\sigma^3 + 2(3 a - 1)\sigma^2 - \right. & \\
 \left. - 8 a \sigma + 8 a\right) r - \sigma^3 \left(\sigma^3 + 4\sigma^2
 - 16a\right)& < 0. \label{eq:ineq:rA1}
\end{align*}
This inequality corresponds to the stability condition for the equilibria $S_{1,2}$.

\qed

\section{Proofs of Lemma \ref{theorem:gen_lorenz:dissip}
 and the completeness of system \eqref{sys:lorenz-general}} 
\label{appendix:dissip}
Suppose that the Lyapunov function has the form
\begin{equation}\label{lyap_func_1}
 V(x,y,z) = \frac{1}{2} \left[ x^2 + y^2 + (a+1)
 \left(z - \frac{\sigma + r}{a + 1} \right)^2 \right].
\end{equation}

Here $V(x,y,z) \to \infty$ as  $|(x,y,z)| \to \infty$.
For an arbitrary solution ${\bf \rm x}(t) = (x(t),y(t),z(t))$
of system \eqref{sys:lorenz-general} we have
\begin{align*}
 \dot{V}(x,y,z) &= x (-\sigma x + \sigma y - a yz)
 + y (rx - y -xz) + ((a + 1) z - (\sigma +r)) (-z + xy)\\
 &= -\sigma x^2 - y^2 - (a + 1) z^2 + (\sigma + r)z.
\end{align*}

Suppose that $\varepsilon \in \left(0, (a+1)\right)$ and
$c = \min\left\{\sigma, 1, (a+1)-\varepsilon\right\} > 0$.
Then
\begin{align*}
\dot{V}(x,y,z) &= -\sigma x^2 - y^2 -
((a + 1) - \varepsilon) z^2 - \varepsilon z^2 + (\sigma + r)z \\
&= -\sigma x^2 - y^2 - ((a + 1) - \varepsilon) z^2 -
\left(\sqrt{\varepsilon} z-\frac{(\sigma + r)}{2\sqrt{\varepsilon}}\right)^2+
\frac{(\sigma + r)^2}{4\varepsilon} \\
&\leq - c (x^2 + y^2 + z^2) + \frac{(\sigma + r)^2}{4\varepsilon}.
\end{align*}

Suppose that $x^2 + y^2 + z^2 \geq R^2$. Then a positive $\varkappa$ exists such that
\[ \dot{V}(x,y,z) \leq -c R^2 + \frac{(\sigma + r)^2}{4\varepsilon} < -\varkappa \quad \text{for } \quad R^2 >
\frac{1}{c} \frac{(\sigma + r)^2}{4\varepsilon}.\]
We choose a number $\eta > 0$ such that
\[ \left\{(x,y,z)~|~ V(x,y,z) \leq \eta \right\} \supset
\left\{(x,y,z)~|~ x^2 + y^2 + z^2 \leq R^2 \right\}, \]
i.e., the relation $x^2 + y^2 + z^2 \leq R^2$ implies that
\[ x^2 + y^2 + (a+1) \left(z - \frac{\sigma + r}{a + 1} \right)^2 =
x^2 + y^2 + z^2 + a z^2 - 2(\sigma+r)z + \frac{(\sigma+r)^2}{a+1} \leq 2\eta.\]
Since
\[ -2(\sigma+r)z ~ \leq ~ 2(\sigma+r) |z| ~ \leq ~ 2(\sigma+r) R,\]
it is sufficient to choose $\eta > 0$ such that
\[ (a+1)R^2 + 2(\sigma+r)R + \frac{(\sigma+r)^2}{a+1} \leq 2\eta,
\quad \text{i.e.} \quad
\eta \geq \frac{1}{2}(a+1)\left(R+\frac{\sigma+r}{a+1}\right)^2.
\]
Further, we can apply Theorem \ref{theorem:dissip1},
which implies the Lemma.
\qed

\medskip

Using  Lyapunov function \eqref{lyap_func_1},
we can prove the boundedness of solutions of
system \eqref{sys:lorenz-general} for  $t \leq 0$.
Note that
\[
\left(2\left(z - \frac{\sigma + r}{a+1}\right)^2+
\frac{(\sigma+r)^2}{2(a+1)^2}\right) -
\left(z - \frac{\sigma + r}{2 (a+1)} \right)^2 =
\left(z - \frac{3(\sigma + r)}{2 (a+1)} \right)^2 \geq 0,
\]
so the inequality
\[
\left(z - \frac{\sigma + r}{2 (a+1)} \right)^2 \leq
\left(2\left(z - \frac{\sigma + r}{a+1}\right)^2+
\frac{(\sigma+r)^2}{2(a+1)^2}\right)
\]
is satisfied.
This implies that
\begin{align*}
 \dot{V} &\geq 2\sigma\left(-\frac{1}{2}x^2\right) +
 2\left(-\frac{1}{2}y^2\right) -
 2 (a+1)\left[\left(z - \frac{\sigma + r}{a+1}\right)^2+
 \frac{(\sigma+r)^2}{4(a+1)^2}\right] +
 \frac{1}{4}\frac{(\sigma+r)^2}{a+1} \\
 &= 2\sigma\left(-\frac{1}{2}x^2\right) + 2\left(-\frac{1}{2}y^2\right) +
 4 \left(-\frac{1}{2} (a+1) \left(z - \frac{\sigma+r}{a+1}\right)^2 \right) -
 \frac{1}{4}\frac{(\sigma+r)^2}{a+1}\\
 &\geq 2\sigma(-V) + 2(-V) + 4(-V) - \frac{1}{4}\frac{(\sigma+r)^2}{a+1}.
\end{align*}
Suppose that $k = 2\sigma + 2 + 4$, and $m = \frac{1}{4}\frac{(\sigma+r)^2}{a+1}$.
Then
\[
 \dot{V} + k V \geq -m.
\]
This implies that
\[
 \frac{d}{dt}(e^{kt}V) = e^{kt}\dot{V} + ke^{kt}V \geq -e^{kt}m.
\]
Thus for $t \leq 0$ we have
\[
 V(0) - e^{kt}V(t) \geq (m e^{kt} - m)/k
\]
or
\[
 V(t) \leq e^{-kt} V(0) + (m e^{-kt} - m)/k.
\]
This implies that $V$ does not tend to infinity
in a finite negative time.
Therefore, any solution $\left(x(t),y(t),z(t)\right)$ of system
\eqref{sys:lorenz-general} does not tend to infinity
in a finite negative time. Thus, differential equation
\eqref{sys:lorenz-general} generates a dynamical system
for $t \in \mathbb{R}$.

\section{Computation of Lyapunov exponents and Lyapunov dimension using MATLAB}\label{appSVD}

The singular value decomposition (\emph{SVD}) of a fundamental matrix ${\rm X}(t)$ has the from
\[
{\rm X}(t)={\rm U}(t){\rm \Sigma}(t){\rm V}^{T}(t): \quad
{\rm U}(t)^T{\rm U}(t) \equiv I \equiv {\rm V}(t)^T{\rm V}(t),
\]
where ${\rm \Sigma}(t)=\text{\rm diag}\{\sigma_1(t),...,\sigma_n(t)\}$
is a  diagonal matrix with positive real diagonal entries
known as \emph{singular values}.
The singular values are the square roots of the eigenvalues
of the matrix ${\rm X}(t)^{*}{\rm X}(t)$ (see \cite{HornJ-1994-book}).
Lyapunov exponents are defined as the upper bounds of the
exponential growth rate of the singular values
of the fundamental matrix of linearized system (see eq.~\eqref{defLE}).

We now give a MATLAB implementation of the discrete SVD method
for computing Lyapunov exponents based on the
product SVD algorithm (see, e.g., \cite{Stewart-1997,DieciL-2008}).

{
\beginMatlab{\textbf{productSVD.m} -- product SVD algorithm}
\fontsize{8}{8}\selectfont
\begin{lstlisting}
function [U, R, V] = productSVD(initFactorization, nIterations)
% Parameters:
%   initFactorization - array containing factor matrices of the
%                       fundamental matrix X, such that:
%          X = initFactorization(:,:,1) * ... * initFactorization(:,:,end);
%   nIterations - number of iterations in the product SVD algorithm.

% dimOde - dimension of the ODEs, nFactors - number of factor matrices
[~, dimOde, nFactors] = size(initFactorization);

% A - 2D array of matrices storing the factor matrices at each iteration
A = zeros(dimOde, dimOde, nFactors, nIterations);
A(:, :, :, 1) = initFactorization;

% Q - array of matrices storing orhogonal matrices of the QR decomposition
Q = zeros(dimOde, dimOde, nFactors+1);

% U, V - orthogonal matrices in the SVD decomposition
U = eye(dimOde); V = eye(dimOde);

% R - array of upper triangular factor matrices, such that after
% the last iteration \Sigma = R(:,:,1) * ... * R(:,:,end)
R = zeros(dimOde, dimOde, nFactors);

% Main loop
for iIteration = 1 : nIterations
    Q(:, :, nFactors + 1) = eye(dimOde, dimOde);
    for jFactor = nFactors : -1 : 1
        C = A(:, :, jFactor, iIteration) * Q(:, :, jFactor+1);
        [Q(:, :, jFactor), R(:, :, jFactor)] = qr(C);
        for kCoord = 1 : dimOde
            if R(kCoord, kCoord, jFactor) < 0
                R(kCoord, :, jFactor) = -1 * R(kCoord, :, jFactor);
                Q(:, kCoord, jFactor) = -1 * Q(:, kCoord, jFactor);
            end;
        end;
    end;

    if mod(iIteration, 2) == 1
        U = U * Q(:, :, 1);
    else
        V = V * Q(:, :, 1);
    end

    for jFactor = 1 : nFactors
        A(:, :, jFactor, iIteration + 1) = R(:, :, nFactors-jFactor+1)';
    end
end

end
\end{lstlisting}
}

{
\beginMatlab{\textbf{computeLEs.m} -- computation of the Lyapunov exponents}
\fontsize{8}{8}\selectfont
\begin{lstlisting}
function LEs = computeLEs(extOde, initPoint, tStep, ...
                                nFactors, nSvdIterations, odeSolverOptions)
% Parameters:
%   extOde - extended ODE system (system of ODEs + var. eq.);
%   initPoint -  initial point;
%   tStep - time-step in the factorization procedure;
%   nFactors - number of factor matrices in the factorization procedure;
%   nSvdIterations - number of iterations in the product SVD algoritm;
%   odeSolverOptions - solver options (sover = ode45);

% Dimension of the ODE :
dimOde = length(initPoint);

% Dimension of the extended ODE (ODE + Var. Eq.):
dimExtOde = dimOde * (dimOde + 1);

tBegin = 0; tEnd = tStep;
tSpan = [tBegin, tEnd];
initFundMatrix = eye(dimOde);
initCond = [initPoint(:); initFundMatrix(:)];

X = zeros(dimOde, dimOde, nFactors);

% Main loop : factorization of the fundamental matrix
for iFactor = 1 : nFactors
    [~, extOdeSolution] = ode45(extOde, tSpan, initCond, odeSolverOptions);

    X(:, :, iFactor) = reshape(...
                        extOdeSolution(end, (dimOde + 1) : dimExtOde), ...
                                                           dimOde, dimOde);
    currInitPoint = extOdeSolution(end, 1 : dimOde);
    currInitFundMatrix = eye(dimOde);

    tBegin = tBegin + tStep;
    tEnd = tEnd + tStep;
    tSpan = [tBegin, tEnd];
    initCond = [currInitPoint(:); currInitFundMatrix(:)];
end

% Product SVD of factorization X of the fundamental matrix
[~, R, ~] = productSVD(X, nSvdIterations);

% Computation of the Lyapunov exponents
LEs = zeros(1, dimOde);
for jFactor = 1 : nFactors
    LEs = LEs + log(diag(R(:, :, jFactor))');
end;
finalTime = tStep * nFactors;
LEs = LEs / finalTime;

end
\end{lstlisting}
}

{
\beginMatlab{\textbf{lyapunovDim.m} -- computation of the Lyapunov dimension}
\fontsize{8}{8}\selectfont
\begin{lstlisting}
function LD = lyapunovDim( LEs )
% For the given array of Lyapunov exponents of a point the function
% compute local Lyapunov dimention of this point.

% Parameters:
%   LEs - array of Lyapunov exponents.

% Initialization of the local Lyupunov dimention :
LD = 0;

% Number of LCEs :
nLEs = length(LEs);

% Sorted LCEs :
sortedLEs = sort(LEs, 'descend');

% Main loop :
leSum = sortedLEs(1);
if ( sortedLEs(1) > 0 )
     for i = 1 : nLEs-1
        if sortedLEs(i+1) ~= 0
           LD = i + leSum / abs( sortedLEs(i+1) );
           leSum = leSum + sortedLEs(i+1);
           if leSum < 0
              break;
           end
        end
    end
end
end
\end{lstlisting}
}

{
\beginMatlab{\textbf{genLorenzSyst.m} -- generalized Lorenz system
\eqref{sys:lorenz-general} along with the variational equation}
\fontsize{8}{8}\selectfont
\begin{lstlisting}
function OUT = genLorenzSyst(t, x, r, sigma, b, a)

% Generalized Lorenz system with
% parameters: r sigma b a

OUT(1) = sigma*(x(2) - x(1)) - a*x(2)*x(3);
OUT(2) = r*x(1) - x(2) - x(1)*x(3);
OUT(3) = -b*x(3) + x(1)*x(2);

% Jacobian at the point [x(1), x(2), x(3)]
J = [-sigma, sigma-a*x(3),  -a*x(2);
    r-x(3),    -1,    -x(1);
    x(2),     x(1), -b];

X = [x(4), x(7), x(10);
    x(5), x(8), x(11);
    x(6), x(9), x(12)];

% Variational equation
OUT(4:12) = J*X;
\end{lstlisting}
}

{
\beginMatlab{\textbf{main.m} -- computation of the Lyapunov exponents and
local Lyapunov dimension for the hidden attractor of generalized Lorenz system
\eqref{sys:lorenz-general}}
\fontsize{8}{8}\selectfont
\begin{lstlisting}
function main

% Parameters of generalized Lorenz system
% that correspond to the hidden attractor
r = 700; sigma = 4; b = 1; a = 0.0052;

% Initial point for trajectory that visualizes the hidden attractor
x0 = [-14.551336132013954 -173.86811769236883 718.92035664071227];

tStep = 0.1;
nFactors = 10000;
nSvdIterations = 3;

% ODE solver parameters
acc = 1e-8; RelTol = acc; AbsTol = acc; InitialStep = acc/10;
odeSolverOptions = odeset('RelTol', RelTol, 'AbsTol', AbsTol, ...
                        'InitialStep', InitialStep, 'NormControl', 'on');

LEs = computeLEs(@(t, x) genLorenzSyst(t, x, r, sigma, b, a), ...
                    x0, tStep, nFactors, nSvdIterations, odeSolverOptions);

fprintf('Lyapunov exponents: %6.4f, %6.4f, %6.4f\n', LEs);

LD = lyapunovDim(LEs);

fprintf('Lyapunov dimension: %6.4f\n', LD);

end
\end{lstlisting}
}

\section{Fishing principle and the existence of a homoclinic orbit in the Glukhovsky--Dolghansky system} \label{appendix:homo}
\subsection{Fishing Principle} 

Consider autonomous system of differential equations \eqref{sys:lorenz-general} with the parameter
\begin{equation} \label{appendix:homo:eq}
	\dot{{\rm x}} = {\bf \rm f}({\bf \rm x},q), \quad t \in \mathbb{R},~
	{\bf \rm x} \in \mathbb{R}^n, \ q\in \mathbb{R}^m.
\end{equation}
Let $\gamma(s), s\in[0,1]$ be a smooth path in the space of the parameter $\{q\}=\mathbb{R}^m$.
Consider the following Tricomi problem \cite{Tricomi-1933}:
{\it Is there a point $q_0\in \gamma(s)$ for which system \eqref{appendix:homo:eq}
with $q_0$ has a homoclinic trajectory?}

Consider system \eqref{appendix:homo:eq} with $q=\gamma(s)$,
and introduce the following notions.
Let ${\rm x}(t,s)^+$ be an outgoing separatrix of the saddle point ${\rm x}_0$
(i.e. $\lim\limits_{t\to-\infty}{\rm x}(t,s)^+={\rm x}_0$)
with a one-dimensional unstable manifold.
Define by ${\rm x}_{\Omega}(s)^+$
the point of the first crossing of separatrix ${\rm x}(t,s)^+$ with the closed set $\Omega$:
$$
{\rm x}(t,s)^+ \, \overline{\in} \, \Omega, \quad t\in(-\infty,T),
$$
$$
{\rm x}(T,s)^+ = {\rm x}_{\Omega}(s)^+ \in \Omega.
$$

If there is no such crossing, we assume that ${\rm x}_{\Omega}(s)^+=\emptyset$ (the empty set).

\begin{theorem}[Fishing Principle \cite{Leonov-2012-PLA,Leonov-2013-IJBC,Leonov-2014-ND}] \label{t1}
Suppose that for the path $\gamma(s)$ there is an $(n-1)$-dimensional bounded manifold $\Omega$ with
a piecewise-smooth edge $\partial\Omega$
that possesses the following properties.
\begin{enumerate}
	\item For any ${\rm x}\in\Omega\setminus\partial\Omega$ and $s\in[0,1]$,
the vector $f({\rm x},\gamma(s))$ is transversal to the manifold $\Omega\setminus\partial\Omega$,
	\label{fish_princ:cond1}

	\item for any $s\in[0,1]$, $f({\rm x}_0,\gamma(s))=0$,
    the point ${\rm x}_0\in\partial\Omega$ is a saddle; \label{fish_princ:cond2}

	\item for $s=0$ the inclusion ${\rm x}_{\Omega}(0)^+\in\Omega\setminus\partial\Omega$
is valid (Fig.~\ref{fig:fish_princ:1}), \label{fish_princ:cond3}

	\item for $s=1$ the relation ${\rm x}_{\Omega}(1)^+=\emptyset$ is valid
	(i.e. ${\rm x}_{\Omega}(1)^+$ is an empty set),
	\label{fish_princ:cond4}

	\item for any $s\in[0,1]$ and ${\rm y}\in\partial\Omega\setminus {\rm x}_0$
there exists a neighborhood $U({\rm y},\delta)=\{{\rm x} \vert \, |{\rm x}-{\rm y}|<\delta\}$
such that ${\rm x}_{\Omega}(s)^+ \, \overline{\in} \, U({\rm y},\delta)$. \label{fish_princ:cond5}
\end{enumerate}
\smallskip

If conditions \ref{fish_princ:cond1})--\ref{fish_princ:cond5}) are satisfied,
then there exists $s_0\in[0,1]$ such that ${\rm x}(t,s_0)^+$ is a homoclinic trajectory
of the saddle point ${\rm x}_0$ (Fig.~\ref{fig:fish_princ:2}).
\end{theorem}

\begin{figure}[!h]
 \centering
 \subfloat[
 {Separatrix ${\rm x}(t,s)^+$, where $s\in[0,s_0)$.}
 ] {
 \label{fig:fish_princ:1}
 \includegraphics[width=0.4\textwidth]{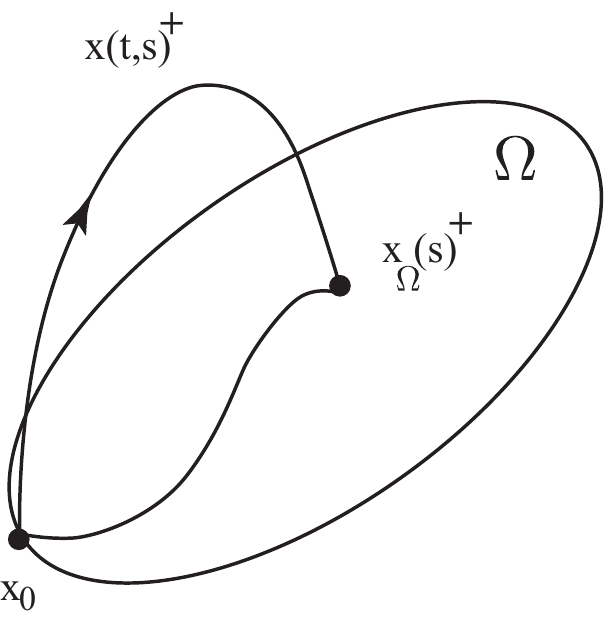}
 }\quad
 \subfloat[
 {Separatrix ${\rm x}(t,s)^+$, where $s=s_0$.}
 ] {
 \label{fig:fish_princ:2}
 \includegraphics[width=0.4\textwidth]{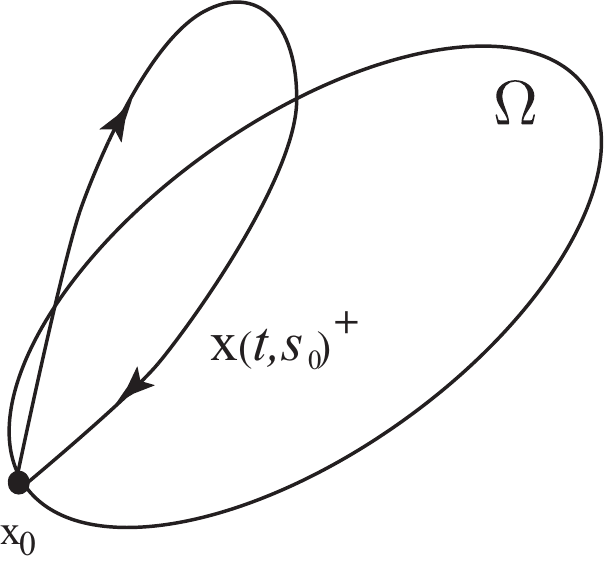}
 }
 \caption{Separatrix ${\rm x}(t,s)^+$, where $s\in[0,s_0]$.}
 \label{fig:fish_princ:1-2}
\end{figure}

The fishing principle can be interpreted as follows.
Fig.~\ref{fig:fish_princ:1} shows a fisherman at the point ${\rm x}_0$ with a
fishing rod ${\rm x}(t, s)^+$.
The manifold $\Omega$ is a lake surface and $\partial\Omega$ is a shore line.
When $s = 0$, a fish has been caught by the fishing
rod. Then, ${\rm x}(t, s)^+, s\in[0,s_0)$ is a path taken by the fishing
rod when it brings the fish to the shore.
Assumption 5) implies that the fish cannot be taken to the
shore $\partial\Omega\backslash {\rm x}_0$, because $\partial\Omega\backslash {\rm x}_0$
is a forbidden zone.
Therefore, only the situation shown in Fig.~\ref{fig:fish_princ:2}
is possible (i.e., at $s = s_0$ the fisherman has caught a fish).
This corresponds to a homoclinic orbit.

Now let us describe the numerical procedure for defining the point $\Gamma$
on the path $\gamma(s)$, which corresponds to a homoclinic trajectory.
Here we assume that conditions \ref{fish_princ:cond1}), \ref{fish_princ:cond2}),
and \ref{fish_princ:cond5}) of the fishing Principle are satisfied.
Consider a sequence of paths
$\gamma_j(s)\subset\{\gamma_{j-1}(s),\,\,s\in[0,1]\}\subset
\{\gamma(s),\,s\in[0,1]\}$, $\forall s\in[0,1]$
such that the length $\{\gamma_j(s)\}$ tends to zero as $j\to +\infty$.
Condition 3) is satisfied for $\gamma_j(0)$ and condition 4) is satisfied
for $\gamma_j(1)$.
This sequence can be obtained if the paths $\gamma$ and $\gamma_j$
are sequentially divided into two paths of the same length
and we choose the path such that for its end points
condition 3) is satisfied and condition 4) is not satisfied (or vice versa).
Obviously, the sequence $\gamma_j(s)$, $s\in[0,1]$
is contracted to the point
$\Gamma\in\{\gamma_j(s),\,s\in[0,1]\},\,\,\forall j$.
This point corresponds to a homoclinic trajectory of system
\eqref{appendix:homo:eq}.

\medskip

Now, consider the conditions of the non-existence of a homoclinic orbit.
Consider the Jacobian matrix of system \eqref{appendix:homo:eq}
$$
  J({\rm x},s)=\frac{\partial {\rm f}}{\partial {\rm x}}({\rm x},\gamma(s)).
$$
Let $\lambda_1 ({\bf \rm x},s,S) \geqslant \cdots \geqslant \lambda_n ({\bf \rm x},s,S)$
denote the eigenvalues of the symmetrized matrix
\begin{equation*}
 \frac{1}{2} \left( S J({\rm x},s) S^{-1} +
 (S J({\rm x},s) S^{-1})^{*}\right),
\end{equation*}
where $S$ is a nonsingular matrix.

Suppose system \eqref{appendix:homo:eq}
has a saddle point ${\rm x}_0 \equiv {\rm x}_0(s),\,\forall s\in[0,1]$,
the point ${\rm x}_0$ belongs to a positively invariant bounded set $K$,
and $J({\rm x}_0,s)$ has only real eigenvalues.

\begin{theorem} [\cite{Leonov-2014-ND}] \label{t2}
Assume that there are a continuously differentiable scalar function $\vartheta({\rm x},s)$
and a nonsingular matrix $S$ such that for system \eqref{appendix:homo:eq} with $q=\gamma(s)$,
the inequality
\begin{equation}\label{ineq:th2}
 \lambda_1 ({\bf \rm x},S) + \lambda_2 ({\bf \rm x},S) +
 \dot{\vartheta}({\bf \rm x}) < 0,
 ~ \forall \, {\bf \rm x} \in K, \quad \forall s \in [0,1]
\end{equation}
is satisfied.
Then system \eqref{appendix:homo:eq} has no homoclinic trajectories for all
$s\in[0,1]$ such that
\begin{equation}\label{6add}
   \lim\limits_{t\to-\infty}{\rm x}(t)=\lim\limits_{t\to+\infty}{\rm x}(t)={\rm x}_0.
\end{equation}
\end{theorem}

\subsection{Existence of a homoclinic trajectories in the Glukhovsky-Dolzhansky system}


Consider the separatrix $x^+(t),y^+(t),z^+(t)$ of the zero saddle point of system \eqref{sys:lorenz-general},
where $x(t)^+>0$, $\forall t\in(-\infty,\tau)$, $\tau$ is a number,
and $\lim\limits_{t\to-\infty}x(t)^+=0$ (i.e. positive outgoing separatrix is considered).

Define the manifold $\Omega$ as
$$
  \begin{aligned}
   \Omega=\{x=0, \, y \leq 0, \, y(\sigma-az) \leq 0, \, y^2+z^2 \leq 2r^2 \}.
  \end{aligned}
$$

\begin{figure}[!h]
 \centering
 \includegraphics[width=0.3\textwidth]{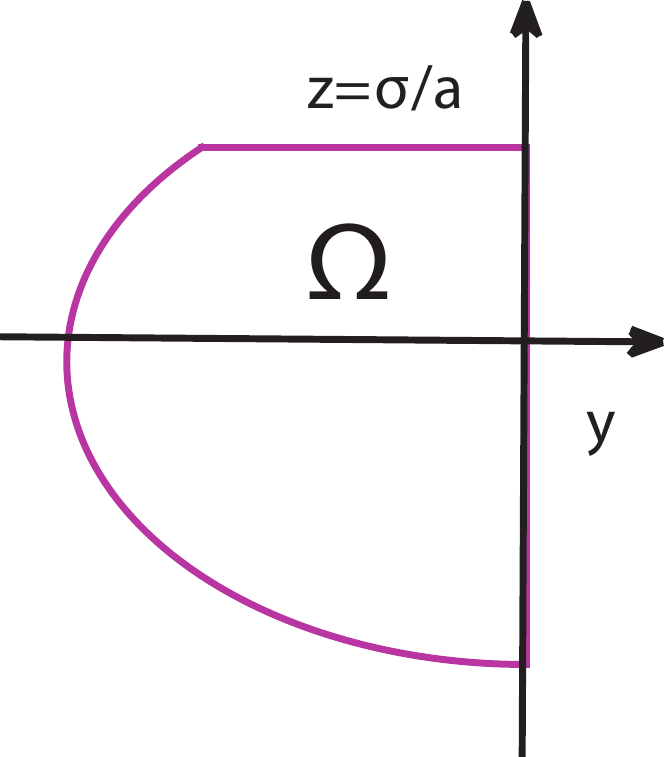}
 \caption{Manifold $\Omega$.}
 \label{fig:fish_princ:3-4}
\end{figure}

\noindent{\bf Check condition 1).}

Inside the set $\Omega \backslash \partial\Omega$ we have
\[
  \dot x= y(\sigma -az) < 0.
\]

\noindent{\bf Check condition 5).}

{\bf a)}  On
\(
  B_1= \{x=0,\, y= 0, \, -\sqrt{2}r \leq z \leq \sigma/a \} 
\)
system \eqref{sys:lorenz-general} has the solution
\[
  		x(t)\equiv y(t)\equiv 0, \quad z(t)=z(0)\exp(-t).
\]

{\bf b)} On
\(
	B_2 =\{x=0, \, y<0, \, z=\sigma\!/a, \, y^2+z^2 \leq 2r^2 \}	
\)
we have
\[
  \ddot x=\sigma y.
\]
\begin{figure}[!h]
 \centering
 \includegraphics[width=0.5\textwidth]{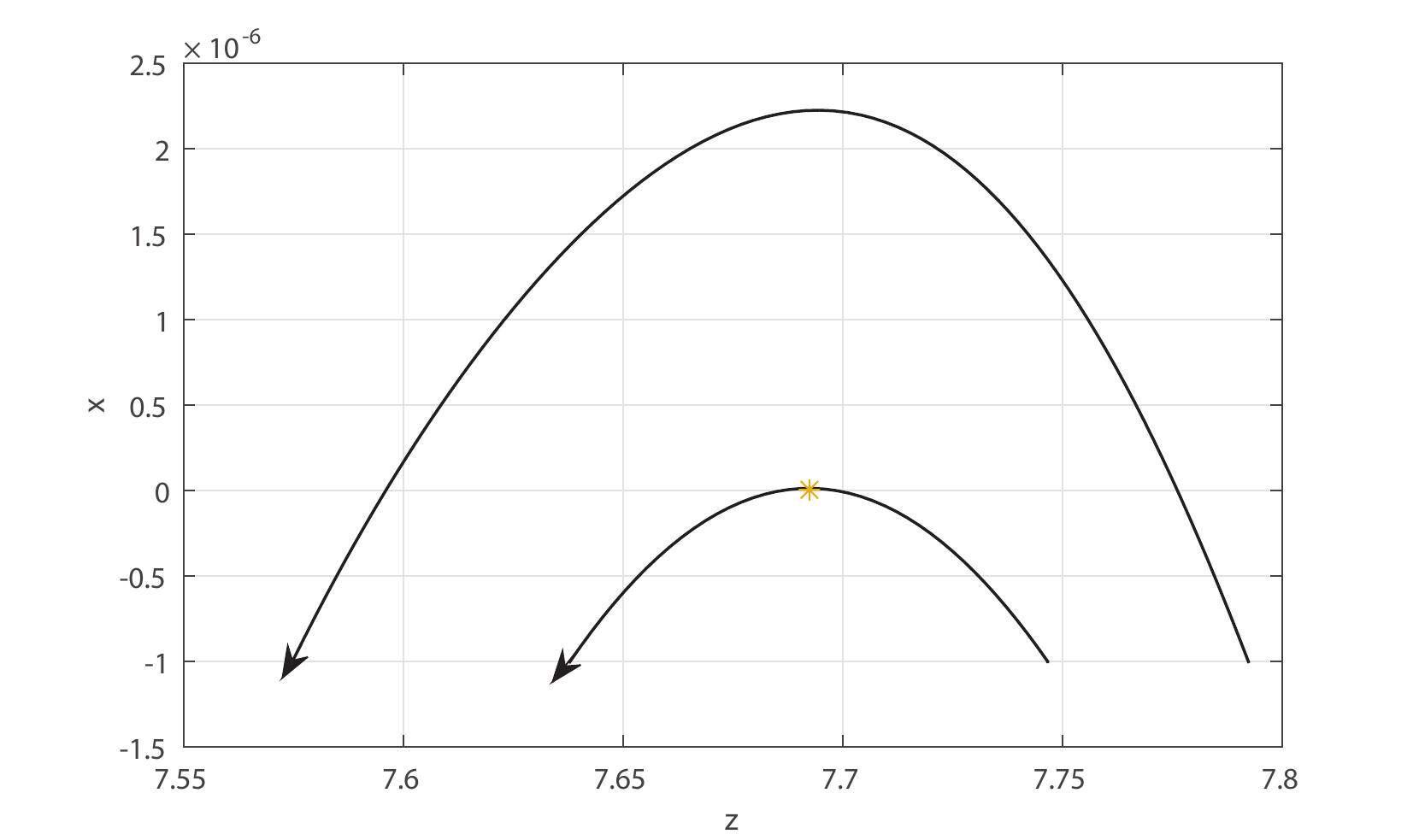}
 \caption{Local behavior of the trajectories of system \eqref{sys:lorenz-general}
 in the neighborhood of set $B_2$  ($\sigma=4$, $a=0.52$, $r=10^5$)}
 \label{fig:fish_princ:5}
\end{figure}
Therefore the local behavior of trajectories in the neighborhood of $B_2$
is shown in Fig.~\ref{fig:fish_princ:5}.

{\bf c)} The set
\(
	B_3 = \{x=0, \, y<0, \, -\sqrt{2}r \leq z \leq \sigma/a, \, y^2+z^2 = 2r^2\} 	
\)
is located outside of the positively invariant cylinder $C$ (see eq.~\eqref{pinv}).
Thus, the separatrices of the zero saddle point (which belongs to $C$) can not reach the set $B_3$.

{\bf Check condition 3)}.

Consider the development of the asymptotic integration of system \eqref{sys:lorenz-general} \cite{Leonov-2015-DAN-GD}.
Assume that
\begin{equation}\label{arseries}
	ar=c-\lambda\varepsilon+O(\varepsilon^2),
\end{equation}
where $c$ and $\lambda$ are some numbers and $\varepsilon=1\!/\sqrt r$ is a small parameter.

\begin{lemma}\label{l10}
For any $\sigma > c, \sigma>1$
there exists a time $T>0$ such that for sufficiently large $r$,
$(x^+(T),y^+(T),z^+(T)) \in \Omega\setminus\partial \Omega$
(i.e. condition \ref{fish_princ:cond3}) of the fishing principle is valid).
\end{lemma}

{\it Proof (sketch).}
Using the transformation
\begin{equation}\label{tr}
	t \to\frac{t}{\sqrt r}, \quad x \to \sqrt r x, \quad y \to r y, \quad z \to rz,
\end{equation}
 we can obtain 
\begin{equation}\label{3.32}
\begin{aligned}
	&\dot x=\sigma y-\varepsilon\sigma x-(c-\lambda\varepsilon+O(\varepsilon^2))yz\\
	&\dot y=x-\varepsilon y-xz \\
	&\dot z=-\varepsilon z+xy.
\end{aligned}	
\end{equation}

{\bf a)}
Consider the zero approximation of system \eqref{3.32}
(system \eqref{3.32}  without $\lambda\varepsilon$ and $\varepsilon=0$)
and its solution $(x_0(t),y_0(t),z_0(t))$.
There are two independent integrals
$$
\begin{aligned}
	& V(x_0(t),y_0(t),z_0(t)) = (\sigma-c)z_0(t)^2+\sigma y_0(t)^2-x_0(t)^2=C_1, \\
	& W = y_0(t)^2+z_0(t)^2-2z_0(t)=C_2.
\end{aligned}
$$
Thus, the positive outgoing separatrix $x_0^+(t),y_0^+(t),z_0^+(t)$ of the saddle point $(x=y=z=0)$
of zero approximation of system \eqref{3.32}
belongs to the intersection of surfaces $V=0$ and $W=0$, i.e.
\begin{equation}\label{3.33}
	V(x_0^+(t),y_0^+(t),z_0^+(t))=0=W(x_0^+(t),y_0^+(t),z_0^+(t)), \forall t \in (-\infty,+\infty).
\end{equation}
From \eqref{3.33} it follows that $x_0^+(t)\ne 0$, $\forall t\in \mathbb{R}^n$ and
$$
 \lim\limits_{t\to+\infty}x_0^+(t)=\lim\limits_{t\to+\infty}y_0^+(t)=
 \lim\limits_{t\to+\infty}z_0^+(t)=0.
$$

\begin{figure}[!ht]
 \centering
 \subfloat[Zero approximation of system \eqref{3.32}:
 integrals ($V$ -- blue, $W$ -- red) and separatrices (black)] {
 \label{fig:gen_lorenz:sprtx:se1}
 \includegraphics[width=0.5\textwidth]{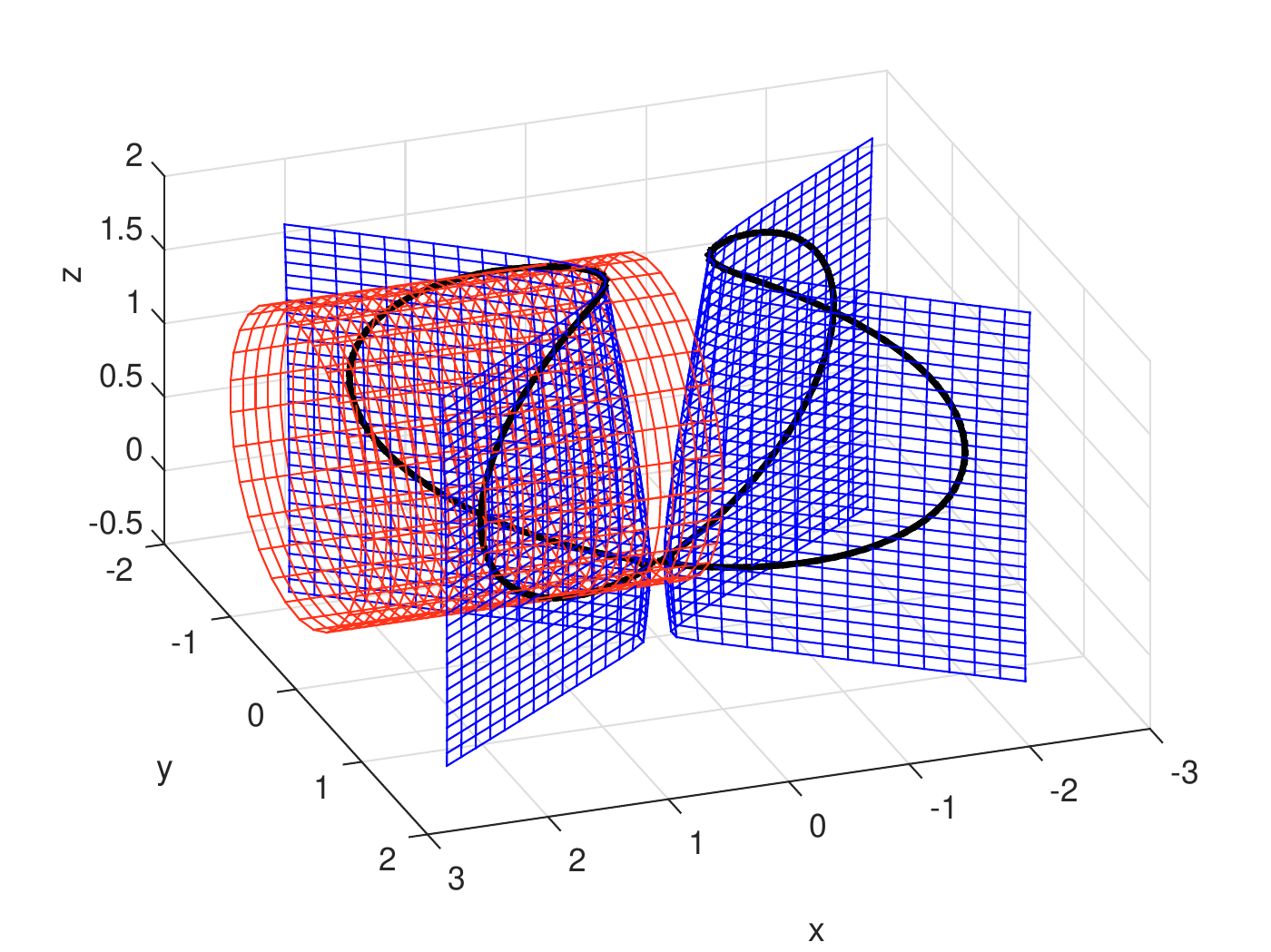}
 }
 \subfloat[
 First approximation of system \eqref{3.32}:
 for $x_1^+(T)=0$ the positive outgoing separatrix is located
 inside magenta domain $\Omega$,
 and near red circuit and the border of blue ellipsoid \eqref{Vforbidden}
 ] {
 \label{fig:gen_lorenz:sprtx:se2}
 \includegraphics[width=0.5\textwidth]{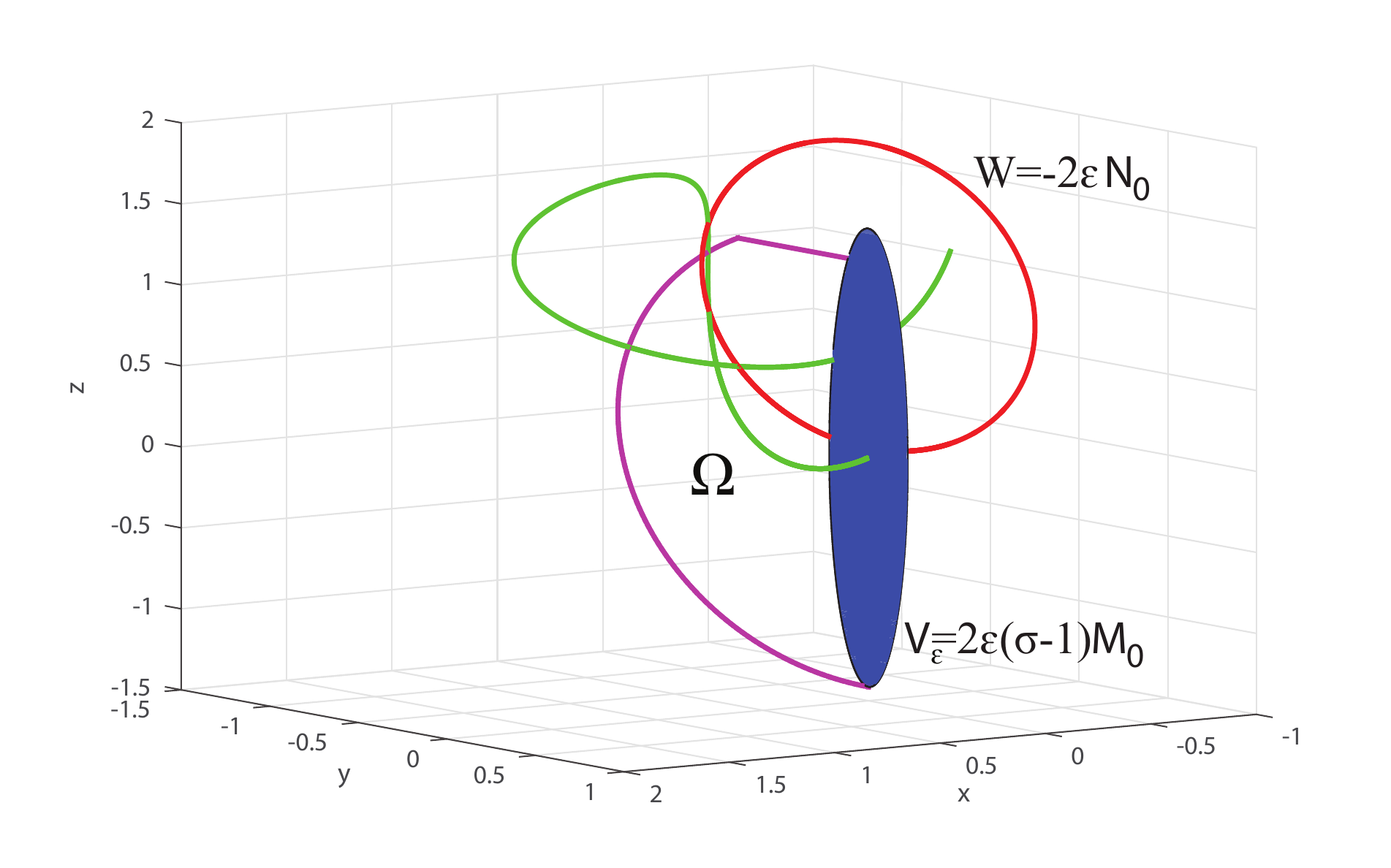}
 }
 \caption{
   $\sigma=2.3445$, $a = 0.0065$, $r = 300$, $c=2$, $\lambda = 1$.
 }
 \label{fig:integrals}
\end{figure}

{\bf b)}
Consider the first approximation of system \eqref{3.32} (system \eqref{3.32} without $O(\varepsilon^2)$).
For the small values of $\varepsilon$  the outgoing separatrix $(x_1^+(t),y_1^+(t),z_1^+(t))$
of the zero saddle point of the first approximation of system \eqref{3.32}
is close  to $(x_0^+(t), y_0^+(t), z_0^+(t))$ on $(-\infty,\tau)$.
Therefore for sufficiently small $\varepsilon$ and some $\tau$
the separatrix $(x_1^+(t), y_1^+(t), z_1^+(t))$
reaches $\delta$-vicinity of the zero saddle.
Then there exists finite $\tau=\tau(\varepsilon,\delta)$ such that
$$
 |x_1^+(\tau(\varepsilon,\delta))|<\delta,\quad |y_1^+(\tau(\varepsilon,\delta))| <
 \delta,\quad |z_1^+(\tau(\varepsilon,\delta))| < \delta.
$$

Consider two functions
\begin{equation}\label{VWe}
\begin{aligned}
	& V_\varepsilon(x,y,z)=(\sigma-c+\lambda\varepsilon)z^2+\sigma y^2-x^2, \quad
	& W = y^2+z^2-2z.
\end{aligned}
\end{equation}
For the derivatives of \eqref{VWe} along the positive outgoing separatrix we have
\begin{equation}\label{dVWe}
 \begin{aligned}
 &
 \frac{d}{dt}V_\varepsilon (t) \equiv \frac{d}{dt}V_\varepsilon(x_1^+(t),y_1^+(t),z_1^+(t))=
 -2\varepsilon V_\varepsilon(x_1^+(t),y_1^+(t),z_1^+(t))
 +2\varepsilon(\sigma-1)x_1^+(t)^2, \\
 &
 \frac{d}{dt}W(t) \equiv \frac{d}{dt}W(x_1^+(t),y_1^+(t),z_1^+(t))=-2\varepsilon W(x_1^+(t),y_1^+(t),z_1^+(t))
  -2\varepsilon z_1^+(t).
 \end{aligned}
\end{equation}
Integrating \eqref{dVWe} from $-\infty$ to $\tau$ and taking into account
\[
  \lim\limits_{\tau_0 \to -\infty} V_\varepsilon(\tau_0)=\lim\limits_{\tau_0 \to -\infty} W(\tau_0)=0,
\]
we obtain
\begin{equation}\label{3.34}
	V_\varepsilon(\tau) =
    2\varepsilon(\sigma-1)\int\limits_{-\infty}^\tau e^{-2\varepsilon (\tau-s)} (x_1^+(t))^2\,dt
    = 2\varepsilon (\sigma-1) M_0 + o(\varepsilon),
\end{equation}
\begin{equation}\label{3.35}
	W(\tau)=
    -2\varepsilon\int\limits_{-\infty}^\tau e^{-2\varepsilon (\tau-s)} (z_1^+(s))^2\,ds
    = -2\varepsilon N_0 + o(\varepsilon),
\end{equation}
where $M_0$ and $N_0$ are some positive numbers.
If $z_1$ and $z_2$ satisfy
\begin{equation}\label{Vforbidden}
   V_\varepsilon(0,0,z_1) = 2\varepsilon (\sigma-1) M_0, \quad W(0,0,z_2)=-2\varepsilon N_0,
\end{equation}
then
\begin{equation}\label{zz}
   z_1 = \sqrt\varepsilon\sqrt{\frac{2M_0(\sigma-1)}{\sigma-c+\lambda\varepsilon}},\quad
   z_2= \varepsilon \frac{2N_0}{1+\sqrt{1-2N_0\varepsilon}}.
\end{equation}
Hence the situation shown in Fig.~\ref{fig:gen_lorenz:sprtx:se2} occurs in
the neighborhood of the saddle point $(x=y=z=0)$ for the surfaces
\begin{equation}\label{forbidden}
\begin{aligned}
 V_\varepsilon(x,y,z) = 2\varepsilon(\sigma-1)M_0,\quad
 W(x,y,z)=2\varepsilon N_0.
\end{aligned}
\end{equation}
The separatrix $(x_1^+(\tau),y_1^+(\tau),z_1^+(\tau)$
has to be near the surfaces $V_\varepsilon(x,y,z)$ and $W(x,y,z)$.
From the mutual disposition of surfaces \eqref{forbidden}
and different order of smallness in \eqref{zz}
it follows that if $x_1^+(\tau)>0$, then $y_1^+(\tau)<0$
and $\dot x_1^+(\tau)<0$ for sufficiently small $\varepsilon$.
Moreover, $\dot x_1^+(t)<0$ for $t>\tau$ and $x_1^+(t) \leq 0$.
This implies the existence of $T>\tau$ such that $x_1^+(T)=0$.
We can obtain similar results for the case $x_1^-(\tau)<0$ (then $y_1^-(\tau)>0$).
The behavior of separatrix $(x^+(\tau),y^+(\tau),z^+(\tau)$
is consistent with the behavior of $(x_1^+(\tau),y_1^+(\tau),z_1^+(\tau)$.

{\bf Check condition 4)}.
We now check condition \ref{fish_princ:cond4}) for system \eqref{sys:lorenz-general}
with parameters \eqref{sys:glukh-dolzh:change_var:param}.
For this system it was proved \cite[pp. 276--277, 269--272]{BoichenkoLR-2005} that if
\begin{equation}\label{3.36}
	R_c<\frac{4(\sigma+1)}{1+\sqrt{1+8 A_c (\sigma+1)}},
\end{equation}
then condition \eqref{ineq:th2} of Theorem~\ref{t2} is satisfied for $S=I, V(x,s)\equiv 0$.

Now we can show that if \eqref{3.36} holds, then condition \ref{fish_princ:cond4})
for system \eqref{sys:lorenz-general} is also satisfied.
Consider the path
\begin{equation}\label{fppath}
\begin{aligned}
	& R_c(s),\quad A_c(s),\quad \sigma(s)\equiv \sigma,\\
	& A_c(0)=0,\quad A_c(1)=A_c,\\
	& R_c(s)\in\left(\frac{2\sigma}{1+\sqrt{1+4A_c(s)\sigma}},\frac{4(\sigma+1)}{1+\sqrt{1+8A_c(s)(\sigma+1)}}\right),\\
	& R_c(0)=\sigma(1+\delta),
\end{aligned}
\end{equation}
where $\delta$ is a small positive number.

For $s=0$ condition \ref{fish_princ:cond4}) is satisfied
(see, e.g., \cite{Leonov-2012-PLA,Leonov-2015-PLA}).
If for some $s_1\in[0,1]$ condition \ref{fish_princ:cond4}) is not satisfied,
then condition \ref{fish_princ:cond3}) is satisfied for $s_1$.
In this case Theorem~\ref{t1} implies that there exists $s_2 \in [0,s_1]$ for which a homoclinic trajectory exists.
But $R_c(s)$ is chosen in such a way that conditions of Theorem \ref{t2} are valid and hence
the homoclinic trajectories do not exist.
This contradiction proves the fulfillment of condition \ref{fish_princ:cond4}) of Theorem~\ref{t1}
for all $s\in[0,1]$.
Condition~\ref{fish_princ:cond4}) is checked.

{\bf Check condition 2)}.
From \eqref{fppath} it is obvious that condition~2 is satisfied.

\begin{remark}
  For $c=\sigma$ Lemma~\ref{l10} is not valid
  since a positive outgoing separatrix of the zero approximation of system \eqref{3.32}
  follows a heteroclitic orbit
  $$
    \lim\limits_{t\to+\infty}x^0_0(t)=\lim\limits_{t\to+\infty}y^0_0(t)=
    \lim\limits_{t\to+\infty}z^0_0(t)=2.
  $$
  In this case we may consider a sequence of systems close to \eqref{sys:lorenz-general}.
  For example, instead of \eqref{sys:glukh-dolzh:change_var:param} we can consider
  $a=a(\beta_k)= \frac{A_c \sigma(\sigma-\beta_k)}{(R_c A_c + 1)^2}$,
  where  $\beta_k$ are a small positive numbers and $\lim\limits_{k \to +\infty} \beta_k =0$,
  such that path \eqref{fppath} satisfies condition 4) of the fishing principle.
  Then, using Lemma~\ref{l10} and the fishing principle,
  we get the sequences of $r^h_k$ and corresponding homoclinic orbits.
  Choosing a convergent subsequence from  $r^h_k$ and using Arzela--Ascoli theorem,
  we can justify the existence of a homoclinic orbit in the initial system with $a(0)$.

  Note also that since  $a$ and $r$  are varying in the asymptotic integration,
  $\partial \Omega$ is also varying.
\end{remark}

\bigskip

Finally we get the following
\begin{theorem}\label{t3} For any fixed $A_c>0, \sigma>1$ there exists a number
$$
  R_c\in\left(\frac{2\sigma}{1+\sqrt{1+4 A_c\sigma}},+\infty\right)
$$
such that system \eqref{sys:lorenz-general}
with parameters \eqref{sys:glukh-dolzh:change_var:param}
has a homoclinic trajectory of the zero saddle point.
\end{theorem}

\section*{Acknowledgments}
This  work  was supported by Russian Scientific Foundation (project 14-21-00041)
and Saint-Petersburg State University.
\bibliographystyle{spmpsci} 

\end{document}